%% file: RauKer_SNR2020_proceedings.tex
\documentclass[submission,copyright,creativecommons]{eptcs}
\providecommand{\event}{SNR 2020} 
\usepackage{breakurl}             

\usepackage{graphicx, algorithm, algorithmic}

\usepackage{tabularx,dblfloatfix}

\input{defs_andreas.tex}
\newtheorem{definition}{Definition}[section]
\newtheorem{remark}{Remark}[section]
\newtheorem{theorem}{Theorem}[section]
\newtheorem{proof}{Proof}[section]

\newcommand{\fod}{\mathcal{D}}

\usepackage{commath}
\usepackage{tikz,pgf,tikz-3dplot}
\usetikzlibrary{calc,arrows,positioning,automata,math}

\tikzset{
	>=stealth',
	punkt/.style={
		rectangle,
		rounded corners,
		draw=black, very thick,
		minimum height=3cm,
		text centered},
	pil/.style={ ->, thick, shorten <=2pt, shorten >=2pt,}
}
\newdimen\xCoord
\newdimen\yCoord
\newcommand*{\ExtractCoordinate}[1]{\path (#1); \pgfgetlastxy{\xCoord}{\yCoord};}

\title{Verification and Reachability Analysis of Fractional-Order Differential Equations Using Interval Analysis}
\author{Andreas Rauh 
\institute{Lab-STICC\\
	ENSTA Bretagne\\
	29806 Brest, France}
\email{Andreas.Rauh@interval-methods.de}	
\and Julia Kersten
\institute{University of Rostock\\
Chair of Mechatronics\\
Justus-von-Liebig-Weg 6, D-18059 Rostock, Germany}
\email{Julia.Kersten@uni-rostock.de}
}
\def\titlerunning{Verification and Reachability Analysis of Fractional-Order Differential Equations}
\def\authorrunning{A.~Rauh \& J.~Kersten}
\begin{document}
\maketitle

\begin{abstract}
Interval approaches for the reachability analysis of initial value problems for sets of classical ordinary differential equations have been investigated and implemented by many researchers during the last decades. However, there exist numerous applications in computational science and engineering, where continuous-time system dynamics cannot be described adequately by integer-order differential equations. Especially in cases in which long-term memory effects are observed, fractional-order system representations are promising to describe the dynamics, on the one hand, with sufficient accuracy and, on the other hand, to limit the number of required state variables and parameters to a reasonable amount. Real-life applications for such fractional-order models can, among others, be found in the field of electrochemistry, where methods for impedance spectroscopy are typically used to identify fractional-order models for the charging/discharging behavior of batteries or for the dynamic relation between voltage and current in fuel cell systems if operated in a non-stationary state. This paper aims at presenting an iterative method for reachability analysis of fractional-order systems that is based on an interval arithmetic extension of Mittag-Leffler functions. An illustrating example, inspired by a low-order model of battery systems concludes this contribution.
\end{abstract}

\section{Introduction}
Floating-point simulation procedures for fractional-order systems are investigated in many current research projects. Such techniques involve --- depending on the system structure --- the numerically efficient and accurate evaluation of functions of the Mittag-Leffler type (basically, if linear models are of interest), the implementation of numerically efficient and robust simulation routines based on the Gr\"unwald-Letnikov differentiation operator (for both linear and nonlinear applications) as well as Laplace domain representations. The latter approach has led to the development of frequency domain-based procedures for linear dynamic systems that are especially applicable to system identification, path planning, and feedback control synthesis for fractional-order system models. A comprehensive software equally applicable for these tasks is \textsc{crone toolbox}, covering time-domain as well as frequency-domain techniques to solve the aforementioned tasks~\cite{MALTI2015769,Podlubny,Oustaloup,Gorenflo:2014,Garrappa_2015,Gorenflo_2002}.

Besides the control-oriented \textsc{crone toolbox}, various numerical schemes for the approximation of fractional-order derivatives and a corresponding numerical integration of systems of fractional-order differential equations have become useful for a variety of engineering applications. Here, especially the {Gr\"unwald-Letnikov operator} that is based on a temporal series expansion of the fractional derivative operator is well-known and serves as the basis for studying the dynamics of engineering applications modeled by fractional-order state equations. Such applications can be found, for example, in simulation, control design, its validation and verification, the design of state-of-charge estimators and model-based aging detection for battery systems. The solution of these tasks relies on 
\begin{enumerate}
	\item deriving state equations for an electric equivalent circuit model in which mostly capacitive elements are generalized by corresponding fractional-order constant phase elements, 
	\item identifying suitable system parameters by means of impedance spectroscopy, and finally
	\item evaluating the resulting system models numerically for a certain finite time horizon~\cite{ZOU2018286,wang,Rauh_IFAC_2020}.
\end{enumerate}

Here, impedance spectroscopy techniques (which are based on exciting the real-life system by harmonic input signals) help to reveal fractional-order behavior which is characterized by long-term memory effects as well as by frequency domain characteristics with amplitude responses that deviate from the classical integer multiples of $\pm20\,\mathrm{dB}$ per frequency decade and, respectively, deviate from integer multiples of $\pm\frac{\pi}{2}$ in the limit values of their corresponding phase response for large angular frequencies.

The long-term memory effect of fractional-order system models also leads to a drawback concerning numerical simulations based on temporal series expansions such as the Gr\"unwald-Letnikov operator. In such cases, it is necessary to account for a large number of summands in the series expansion to capture the long-term memory effects with sufficient accuracy. Moreover, restarting a time-domain evaluation of fractional-order differential equations purely on the basis of current state information is prone to significant truncation errors. Hence, a significant amount of work has also been invested in the derivation of suitable initialization procedures as well as in the quantification of the respective truncation errors.

In many engineering applications, the number of summands that are employed during a series-based simulation is, however, often restricted to a quite small value. This inherently leads to temporal truncation errors which are, unfortunately, in many cases not quantified in a rigorous manner. This fact becomes even more challenging when system models with uncertain parameters are accounted for. Assuming that for some application at hand the influence of truncation errors can be neglected, a naive generalization of the Gr\"unwald-Letnikov operator-based series expansions to systems with interval parameters would still inevitably lead to unreasonably wide interval enclosures. This is mostly caused by the mutual couplings between various state variables which --- if characterized by finitely large interval boxes in a multi-dimensional state space --- would have to be propagated through each of the summands in the series expansion and therefore provoke the wrapping effect that is also well-known from solving initial value problems for sets of classical ordinary differential equations~\cite{Ned:2006}. 

However, if it can be proven beforehand that the state equations under investigation possess the property of asymptotic stability, quasi-analytic representations of the enclosure of the systems' time responses can be defined by using Mittag-Leffler functions~\cite{Rauh_mmar_2019,Rauh_ECC_2020}. Those functions represent a generalization of classical exponential-type solutions which were originally developed for integer-order differential equations with the goal to suppress the wrapping effect in a computationally inexpensive manner for system models with a-priori proven stability properties.

For further general discussions about modeling dynamic systems by using fractional-order representations, their identification, as well as their potential use in control engineering tasks, the reader is referred to~\cite{Oustaloup}. Moreover, the paper~\cite{Rauh_Acta_2020} provides an overview of the state-of-the-art concerning already existing approaches for a verified simulation of fractional-order models. These include the definition of differential inclusions, the exploitation of cooperativity (i.e., monotonicity of solutions with respect to initial conditions) and positivity of state variables in simulation as well as state estimation, nonlinear time transformations to cast fractional-order models into equivalent integer-order representations, and, finally, also the extension of Picard iteration procedures~\cite{Lyons} to non-integer-order models.

Following the summary of required preliminaries and the problem formulation in Sec.~\ref{sec:prelim} as well as an exemplary comparison between low-dimensional fractional-order system models and significantly higher-dimensional approximations in terms of integer-order formulations in Sec.~\ref{sec:lin_frac_comparison}, a Mittag-Leffler type state enclosure technique derived in~\cite{Rauh_Acta_2020,Rauh_ECC_2020} for fractional-order models is reviewed in Sec.~\ref{sec:exponential} together with specific computational aspects for an interval-based evaluation of Mittag-Leffler type functions in Sec.~\ref{sec:interval_MLF}. In contrast to previous work of the authors, where scalar differential equations were in the focus of mostly theoretical applications, Sec.~\ref{sec:example} aims at applying the suggested iteration procedure to a non-scalar, practically motivated simplified dynamic battery model with uncertain parameters. Finally, conclusions and an outlook on future work are given in Sec.~\ref{sec:concl}.

\section{Preliminaries}\label{sec:prelim}
%
In this paper, a verified simulation routine for fractional-order system models that is based on a differential formulation of the Picard iteration is employed. It exploits Mittag-Leffler functions to represent guaranteed state enclosures for which the free parameters (describing a kind of decay rate of the initial conditions toward the equilibrium) are computed in an iterative manner.

Throughout this paper, Mittag-Leffler functions~\cite{Haubold,Gorenflo:2014,Garrappa_2015} are given by their two-parameter representation
\begin{equation}\label{eq:Mittag_Leffler}
E_{\nu,\beta}(\zeta) = \sum\limits_{i=0}^\infty
\frac{\zeta^i}{\Gamma\rb{\nu i + \beta}}
\end{equation}
with the general argument $\zeta \in \mathbb{C}$ as well as the parameters $\nu \in \mathbb{R}_+$ and $\beta \in \mathbb{R}$. Note, although this function is defined for each point $\zeta$ in the complex plane, we will restrict the discussion to the case of real-valued arguments and outputs of the Mittag-Leffler function.
In addition, assume that all dynamic systems in this paper are given in terms of explicit, autonomous, time-invariant\footnote{Note, the restriction to autonomous, time-invariant systems can be removed by the introduction of auxiliary state variables for the time argument as well as for time- and state-dependent expressions included in control inputs. Corresponding procedures, leading to an increase of the system dimension, were discussed, for example, in~\cite{Rauh_SCAN14} for the integer-order case.} state equations which are re-written according to Def.~\ref{def:quasi_lin} into a quasi-linear form. This reformulation enhances the efficiency of the numerical evaluation of state enclosures, especially with respect to suppressing the wrapping effect.

\begin{definition}[\label{def:quasi_lin}Quasi-linear fractional-order system model of Caputo type]
	A commensurate-order quasi-linear set of fractional-order differential equations of Caputo type~\cite{Oustaloup,Podlubny} is defined by
	\begin{equation}\label{eq:state_space_frac}
	\begin{split}
	{}_{t_0=0}\mathcal{D}_t^\nu\mathbf{x}(t) = 
	{\mathbf{x}}^{\rb{\nu}}(t) & = 
	{\mathbcal{f}}\rb{{\mathbf{x}}(t)} =: 
	{\mathbcal{A}}\rb{{\mathbf{x}}(t)} \cdot {\mathbf{x}}(t)~,~~
	{\mathbf{x}}(t) \in \mathbb{R}^n~,~~
	{\mathbcal{A}}\rb{{\mathbf{x}}(t)} \in \mathbb{R}^{n\times n}
	~,~~
	0 < \nu <1
	\enspace ,
	\end{split}
	\end{equation}
	where $\nu$ is the non-integer derivative order and ${\mathbcal{f}}\rb{{\mathbf{x}}(t)}$ the nonlinear right-hand side of the system model re-written into a quasi-linear form; initial conditions  $\mathbf{x}(t_0 = 0)$ are specified in terms of the interval bounds
	\begin{equation}\label{eq:init_x}
	\mathbf{x}(0) \in \sqb{\mathbf{x}}(0) = \intv{\ul{\mathbf{x}}(0)}{\ol{\mathbf{x}}(0)}
	\enspace .
	\end{equation}
\end{definition}

In~(\ref{eq:init_x}), the initial state vector ${\mathbf{x}}(0)$ is described by the interval representation $\sqb{\mathbf{x}}(0)$, where the inequalities ${\ul{{x}}_i(0)}\le{\ol{{x}}_i(0)}$ hold element-wise for each vector component $i \in \{1,\ldots,n\}$. 

\begin{definition}[\label{def:diag_dominant}Diagonally dominant model] Diagonally dominant quasi-linear system models are given by the state-space representation
	\begin{equation}\label{eq:diag_dom}
	\begin{split}
	{}_{0}\mathcal{D}_t^\nu\mathbf{z}(t) = 
	{\mathbf{z}}^{\rb{\nu}}(t) 
	& = \mathbf{f}\rb{{\mathbf{z}}(t)} = 
	\rb{{\mathbf{T}}^{-1} \cdot 
		{\mathbcal{A}}\rb{\mathbf{T} \cdot {\mathbf{z}}(t)} \cdot \mathbf{T}} \cdot {\mathbf{z}}(t)  =: 
	{\mathbf{A}}\rb{{\mathbf{z}}(t)} \cdot {\mathbf{z}}(t)
	\end{split}
	\end{equation}
	after a suitable similarity transformation
	\begin{equation}\label{eq:trafo}
	{\mathbf{x}}(t) = 
	{\mathbf{T}} \cdot {\mathbf{z}}(t)~,~~
	{\mathbf{T}} \in \mathbb{R}^{n \times n}~,~~
	{\mathbf{z}}(t) \in \mathbb{R}^{n}
	\end{equation}
	of the system model in Def.~\ref{def:quasi_lin}.
\end{definition}

\begin{remark}
	In this paper, we restrict ourselves to the case of real-valued similarity transformations in~(\ref{eq:trafo}). These transformations lead to the real-valued initial state enclosures 
	\begin{equation}\label{eq:init_cond_trafo}
	\mathbf{z}(0) \in \sqb{\mathbf{z}}(0) =
	\mathbf{T}^{-1} \cdot \sqb{\mathbf{x}}(0)
	\enspace .
	\end{equation}
	According to~\cite{Rauh_SCAN14,Rauh_mmar_2019}, also complex-valued similarity transformations are possible which are advantageous for the case of systems with conjugate-complex eigenvalues and, hence, oscillatory dynamics. For both the real- and complex-valued case with system models having an eigenvalue multiplicity of one, the transformation matrix ${\mathbf{T}}$ is composed of the eigenvectors of ${\mathbcal{A}}\rb{\mathbf{x}_\mathrm{m}}$, computed at the interval midpoint $\mathbf{x}_\mathrm{m} = \frac{1}{2}\cdot \rb{{\ul{\mathbf{x}}(0)}+{\ol{\mathbf{x}}(0)}}$. Transformations applicable to systems with higher eigenvalue multiplicities were derived in~\cite{Rauh_SCAN14} for dynamic models with integer-order derivatives.
\end{remark}

\begin{remark}
	Where it is necessary for a compact notation of the iteration procedure given in the following section, it is further assumed that a translation of the state vector has been performed prior to solving the considered simulation task (in accordance with Defs.~\ref{def:quasi_lin} and~\ref{def:diag_dominant}) so that the state trajectories converge to the origin of the state space if the dynamics are asymptotically stable.
\end{remark}

\section{Comparison Between Linear Fractional-Order Models and Classical Integer-Order Approximations}\label{sec:lin_frac_comparison}

To illustrate fundamental differences between linear fractional-order differential equations and integer-order system representations, consider the dynamic system model
\begin{equation}
y^{(0.5)} (t) = -y(t) + u(t)
\end{equation}
with the input signal $u(t)$ and the system output $y(t)$. As an example, $u(t)$ and $y(t)$ may represent the total current and the voltage drop, respectively, over a parallel connection of an Ohmic resistor and a constant phase element with a maximum phase shift of $-\frac{\pi}{4}$. Such elements appear typically in the domain of electrochemical impedance spectroscopy~\cite{ZOU2018286,wang}. In the frequency domain, these so-called Warburg elements are expressed by 
\begin{equation}\label{eq:ex_frac1}
F(\jmath\omega) = \frac{Y(\jmath\omega)}{U(\jmath\omega)}
= 
\frac{1}{1+(\jmath \omega)^{0.5}}
\enspace. 
\end{equation}

As mentioned in the introduction, parameter identification and numerical simulation of such elements are either possible by using specialized routines for fractional-order system models (as it is considered also in this paper) or by approximating them by finite-dimensional sets of ordinary differential equations. A possibility for doing so is based on firstly approximating the fractional differentiation operator by a suitable transfer function and subsequently transforming this system model back into the time domain, leading to the approximating state-space model
\begin{equation}\label{eq.approx}
\dot{\tilde{\mathbf{z}}}(t) = \mathbf{A} {\tilde{\mathbf{z}}}(t) + \mathbf{b} u(t)
\quad \text{with}
\quad 
y(t) \approx \mathbf{c}^T {\tilde{\mathbf{z}}}(t) + d u(t)
\enspace .
\end{equation}
If the recursive approximation routines developed by A.~Oustaloup, cf.~\cite{Khazali,Oustaloup,817385} as well as the freely available \textsc{Matlab} implementation~\cite{Chen}, are used for that purpose, the approximations in Fig.~\ref{fig:ex_nyqist} are obtained, where the frequency range $\omega \in \intv{0.01}{100}\,\mathrm{\frac{rad}{s}}$ over which the approximation is performed, is highlighted by the vertical dotted lines in Figs.~\ref{fig_ampl} and~\ref{fig_phase}. Here, the black dashed curves denote the exact fractional-order behavior, while the red and blue models correspond to approximations of~(\ref{eq:ex_frac1}) in terms of rational transfer functions of the orders $5$ (red) and $11$ (blue), respectively. 

Exemplarily, the numerical values for the approximating system matrices in~(\ref{eq.approx}) are given by
\begin{equation}
\begin{split}
\setcounter{MaxMatrixCols}{11}
\mathbf{A} &=
\scalebox{0.7}{$
\begin{bmatrix} -559.71 & -277.1 & -76.197 & -24.835 & -9.9316 & -2.4829 & -0.77609 & -0.29764 & -0.067652 & -0.017081 & -0.0039062\\ 256.0 & 0 & 0 & 0 & 0 & 0 & 0 & 0 & 0 & 0 & 0\\ 0 & 128.0 & 0 & 0 & 0 & 0 & 0 & 0 & 0 & 0 & 0\\ 0 & 0 & 32.0 & 0 & 0 & 0 & 0 & 0 & 0 & 0 & 0\\ 0 & 0 & 0 & 8.0 & 0 & 0 & 0 & 0 & 0 & 0 & 0\\ 0 & 0 & 0 & 0 & 4.0 & 0 & 0 & 0 & 0 & 0 & 0\\ 0 & 0 & 0 & 0 & 0 & 1.0 & 0 & 0 & 0 & 0 & 0\\ 0 & 0 & 0 & 0 & 0 & 0 & 0.25 & 0 & 0 & 0 & 0\\ 0 & 0 & 0 & 0 & 0 & 0 & 0 & 0.125 & 0 & 0 & 0\\ 0 & 0 & 0 & 0 & 0 & 0 & 0 & 0 & 0.03125 & 0 & 0\\ 0 & 0 & 0 & 0 & 0 & 0 & 0 & 0 & 0 & 0.0078125 & 0 \end{bmatrix}
$},
\\
\mathbf{b}& =
\begin{bmatrix} 8.0 & 0 & 0 & 0 & 0 & 0 & 0 & 0 & 0 & 0 & 0 \end{bmatrix}^T,\\
\mathbf{c}^T &=
\scalebox{0.8}{$
\begin{bmatrix} 1.7692 & 2.3998 & 1.3484 & 0.77563 & 0.48601 & 0.16983 & 0.066826 & 0.029657 & 0.0073522 & 0.0019502 & 0.00045835 \end{bmatrix}$},\\
d &= 0.030653
\enspace. 
\end{split}
\raisetag{5mm}
\end{equation}
By a comparison of the frequency responses and Nyquist plots for both chosen approximation orders in Fig.~\ref{fig:ex_nyqist}, it can be concluded that all state variables of the integer-order approximation given above have a significant contribution to the approximation of the fractional-order behavior. If lower approximation orders were used, the quality of approximation deteriorates which can be seen by the oscillating nature of the red curves in Fig.~\ref{fig:ex_nyqist}. Obviously, the advantage of working directly with the fractional-order model is the significantly smaller number of required parameters and state variables. 

\begin{figure}[htp]
	\centering
	\subfloat[{\label{fig_ampl}}Magnitude response.]{\resizebox{0.32 \linewidth}{!}{{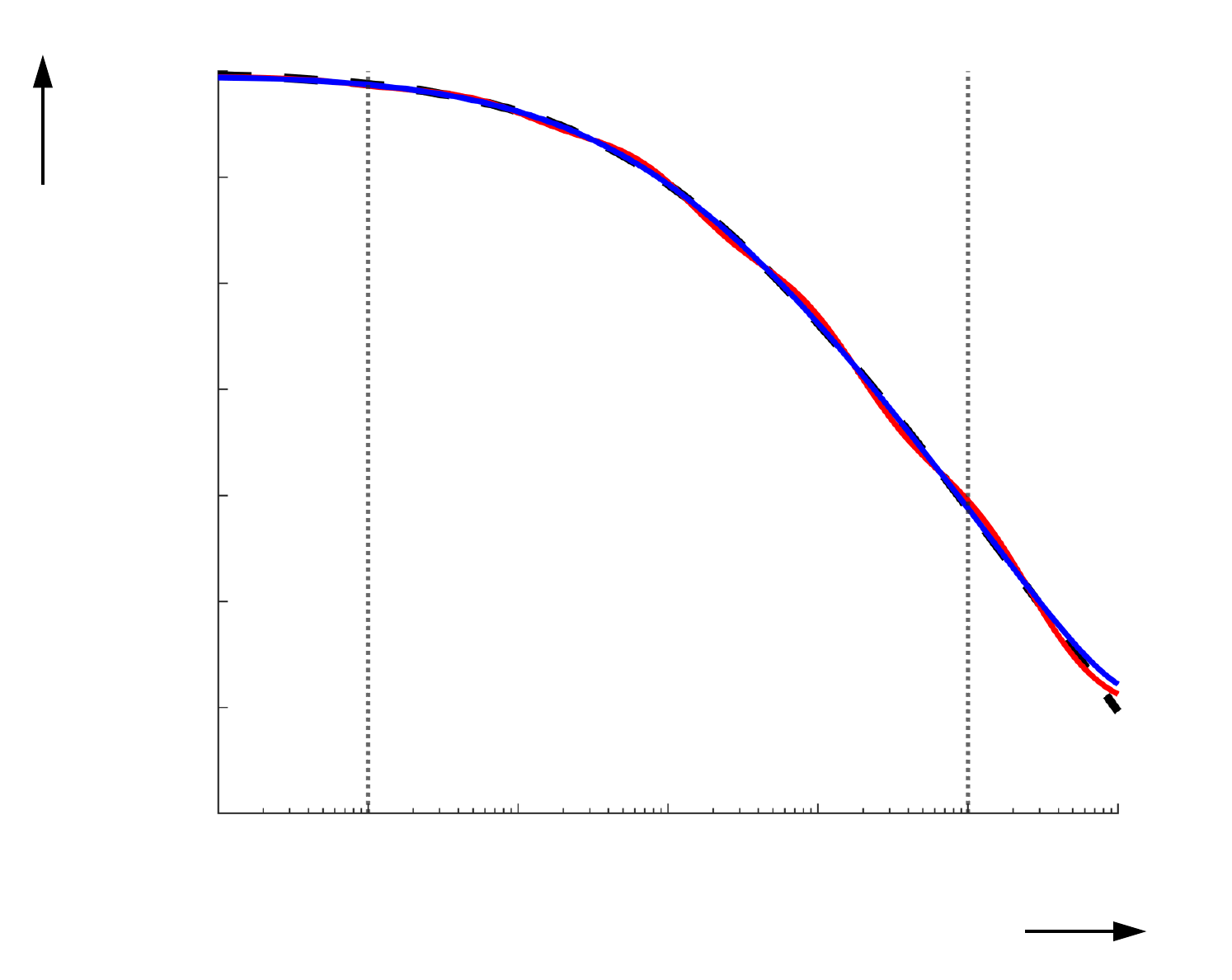}}}
	\hfill
	\subfloat[{\label{fig_phase}}Phase response.]{\resizebox{0.32 \linewidth}{!}{{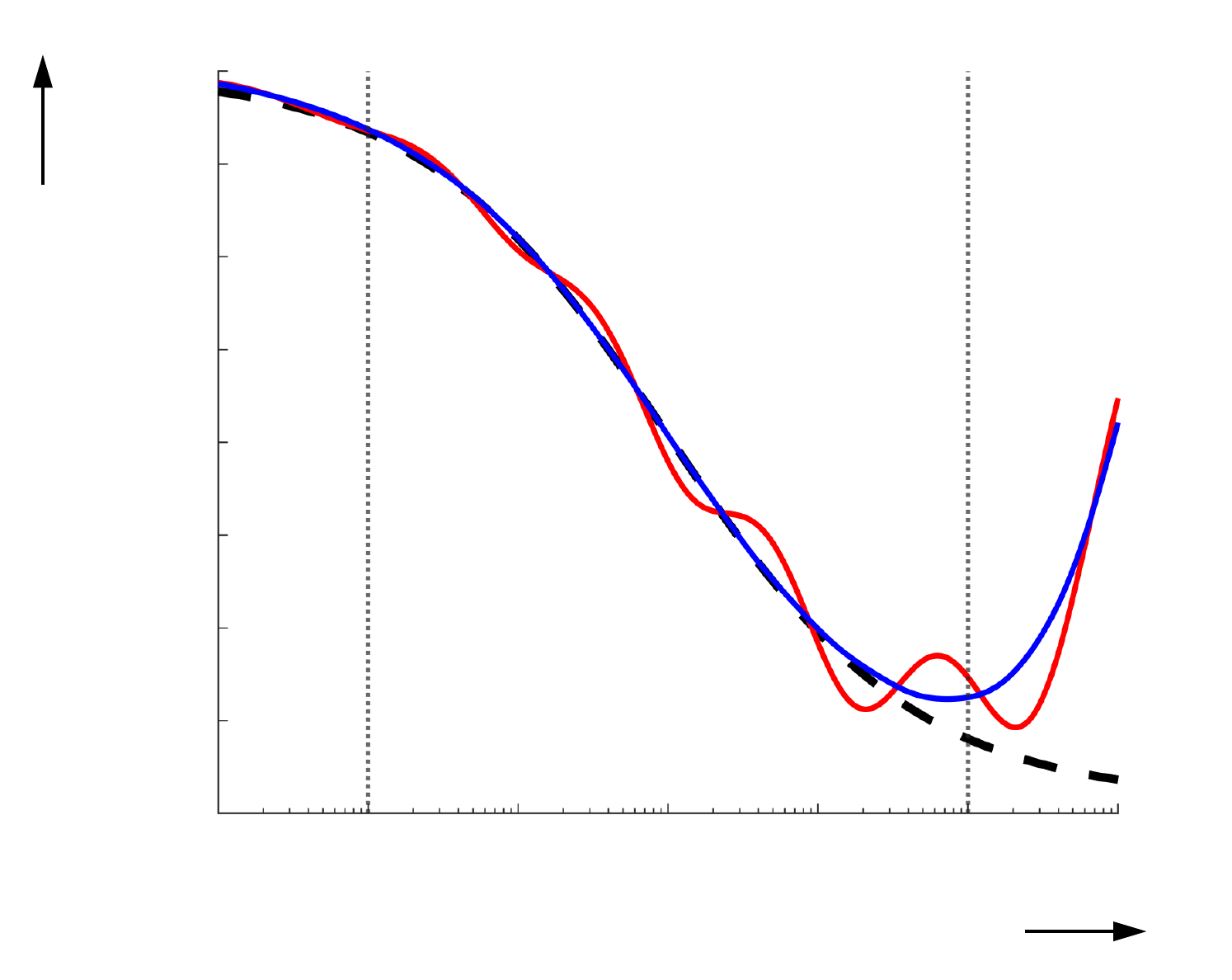}}}
	\hfill
	\subfloat[{\label{fig_nyquist}}Nyquist plot with selected frequencies  $\ul{\omega}=0.01\mathrm{\frac{rad}{s}}$ ($\circ$) and $\ol{\omega}=100\mathrm{\frac{rad}{s}}$ ($\diamond$).]{\resizebox{0.32 \linewidth}{!}{{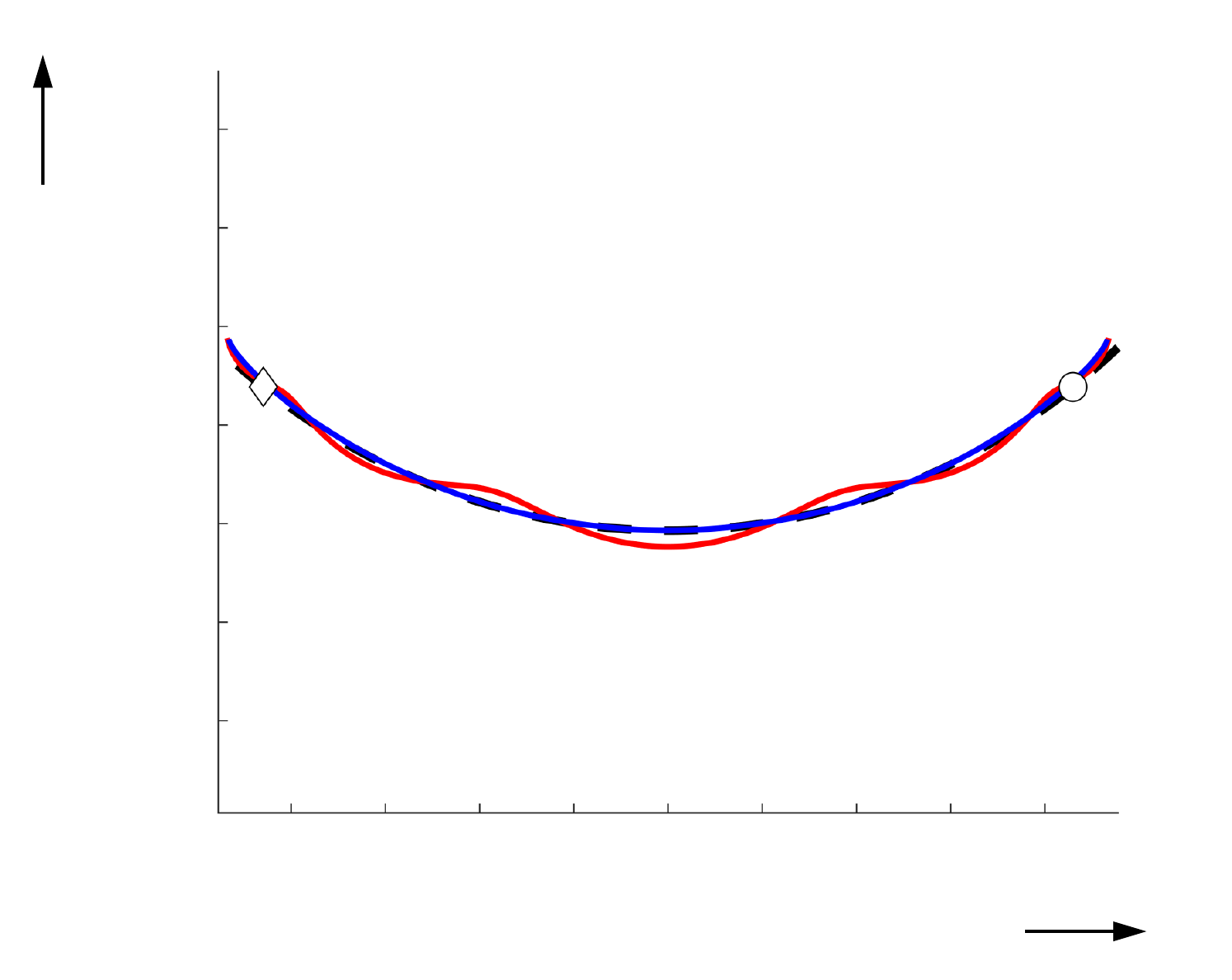}}}
	\caption{{\label{fig:ex_nyqist}}Frequency domain analysis of the fractional-order model~(\ref{eq:state_space_frac}) in comparison with different integer-order approximations (black dashed: fractional-order model; {\color{red}{red solid}}: 5th order approximation; {\color{blue}{blue solid}}: 11th order approximation).}
\end{figure}

\section{Mittag-Leffler Type State Enclosures for Fractional-Order Differential Equations}\label{sec:exponential}
The focus of this section is the summary of a Mittag-Leffler type state enclosure technique for sets of commensurate fractional-order models. The fundamental iteration summarized in the following Theorem~\ref{thm:MLF_encl} was first published by the authors in~\cite{Rauh_mmar_2019} and~\cite{Rauh_ECC_2020} with a further in-depth discussion of the infinite-horizon memory property that becomes crucial as soon as the integration time horizon is divided into several, finitely long temporal subslices in~\cite{Rauh_Acta_2020}.

\begin{definition}[\label{def:MLF_encl}Mittag-Leffler type state enclosure]
	The time-dependent expression
	\begin{equation}\label{eq:sol_MLF}	
	{\mathbf{z}}^*(t) \in 
	\mathbf{E}_{\nu,1}{\left(\left[\boldsymbol\Lambda\right] \cdot t^\nu \right)} \cdot
	\left[{\mathbf{z}}_{e}\right]\left(0\right)
	~,~~
	\mathbf{E}_{\nu,1}{\left(\left[\boldsymbol\Lambda\right] \cdot t^\nu \right)}
	:= \diag{{E}_{\nu,1}{\left(\left[\lambda_i\right] \cdot t^\nu \right)}}
	~,~~
	\sqb{{\mathbf{z}}_{e}}\left(0\right) = \sqb{{{\mathbf{z}}}_0}
	\end{equation}with the diagonal parameter matrix
	$
	\left[\boldsymbol\Lambda\right] := \diag{\left[\lambda_i\right]}$, $i\in \{1,\ldots,n\}$,
	is denoted as a \emph{verified Mittag-Leffler type state enclosure} for the system model~(\ref{eq:diag_dom}) with~(\ref{eq:init_cond_trafo}) if it is determined according to Theorem~\ref{thm:MLF_encl}.
\end{definition}

\begin{theorem}[\cite{Rauh_ECC_2020,Rauh_Acta_2020,Rauh_mmar_2019}~\label{thm:MLF_encl}Iteration for Mittag-Leffler type enclosures]
	The Mittag-Leffler type state enclosure~(\ref{eq:sol_MLF}) is guaranteed to contain the set of all reachable states ${\mathbf{z}}^*(T)$ at the point of time $t=T>0$ according to
	\begin{equation}
	{\mathbf{z}}^*(T) \in 
	\mathbf{E}_{\nu,1}{\left(\left[\boldsymbol\Lambda\right] \cdot T^\nu \right)} \cdot
	\left[{\mathbf{z}}_{e}\right]\left(0\right)\, ,~
	{\mathbf{z}}^*(t) \in \mathbb{R}^{n}
	\enspace ,
	\end{equation}
	if $\left[\boldsymbol\Lambda\right]$ is set to the outcome of the converging iteration 
	\begin{equation}\label{eq:iteration_MLF}
	\left[{\lambda}_i\right]^{\dbr{\kappa +1}}
	:=
	\frac{
		{{f}_i}\left(
		\mathbf{E}_{\nu,1}{\left(\left[\boldsymbol\Lambda\right]^{\dbr{\kappa}} \cdot \sqb{t}^\nu\right)} \cdot
		\left[{\mathbf{z}}_{e}\right]\left(0\right)
		\right)}
	{
		{E}_{\nu,1}{\left(\left[\lambda_i\right]^{\dbr{\kappa}} \cdot \sqb{t}^\nu\right)} \cdot
		\left[{{z}}_{e,i}\right]\left(0\right)
	}
	\enspace , 
	\end{equation}
	$i\in \{1,\ldots,n\}$, with the prediction horizon $\sqb{t} = \intv{0}{T}$.
\end{theorem}

\begin{proof}
	According to~\cite{Dorjgotov,Ghosh}, the exact solution of a linear fractional-order differential equation 
	\begin{equation}\label{eq:sol_frac1}
	z^{(\nu)}(t) = \lambda \cdot z(t)
	\end{equation}
	of Caputo type, for which solely initial state information at the point $t=0$ is available, is given by the analytic expression
	\begin{equation}\label{eq:sol_frac2}
	z(t) = E_{\nu,1}\rb{\lambda t^\nu} \cdot z(0)
	\enspace .
	\end{equation}
	The relation~(\ref{eq:sol_frac2}) serves as an ansatz for describing the verified state enclosures according to Def.~\ref{def:MLF_encl} by substituting it according to~\cite{Rauh_Acta_2020,Rauh_ECC_2020} into 
	the differential formulation of the Picard iteration
	\begin{equation}\label{eq:exp1_fde}
	\begin{split}
	{\mathbf{z}}^{\rb{\nu}}(t) & \in
	\rb{\left[\boldsymbol\Lambda\right]^{\dbr{\kappa +1}}} \cdot
	\mathbf{E}_{\nu,1}{\left(\left[\boldsymbol\Lambda\right]^{\dbr{\kappa +1}} \cdot t^\nu \right)}
	\cdot
	\left[{\mathbf{z}}_{e}\right]\left(0\right)\\
	& =
	{\mathbf{f}}\left(
	\mathbf{E}_{\nu,1}{\left(\left[\boldsymbol\Lambda\right]^{\dbr{\kappa}} \cdot t^\nu \right)} \cdot
	\left[{\mathbf{z}}_{e}\right]\left(0\right)
	\right) =:
	{\mathbf{f}}\left(\left[{\mathbf{z}}_{e}
	\right]^{\dbr{\kappa}}(\sqb{t})
	\right)
	\enspace  .
	\end{split}
	\raisetag{13mm}
	\end{equation} 
	
\noindent Overapproximating the Mittag-Leffler type state enclosure 
	$\mathbf{E}_{\nu,1}{\left(\left[\boldsymbol\Lambda\right]^{\dbr{\kappa+1}} \cdot t^\nu \right)}$ in the iteration step \mbox{$\kappa+1$} by the enclosure ${\left[{\mathbf{z}}_{e}
		\right]^{\dbr{\kappa}}(\sqb{t})}$ obtained in the previous iteration step~ $\kappa$ --- under the assumption
	that \linebreak \mbox{${\left[{\mathbf{z}}_{e}
		\right]^{\dbr{\kappa+1}}(\sqb{t})} \subseteq {\left[{\mathbf{z}}_{e}
		\right]^{\dbr{\kappa}}(\sqb{t})}$} holds due to convergence of the iteration formula ---
	 on the first line of~(\ref{eq:exp1_fde})
	leads to
	\begin{equation}
	{\diag{\sqbs{\tilde{\lambda}_i}^{\dbr{\kappa +1}}}}  \cdot
	{\left[{\mathbf{z}}_{e}
		\right]^{\dbr{\kappa}}(\sqb{t})} =
	{\mathbf{f}}\left(\left[{\mathbf{z}}_{e}
	\right]^{\dbr{\kappa}}(\sqb{t})
	\right)
	\enspace .
	\end{equation} 
	Solving this expression for $\sqbs{\tilde{\lambda}_i}^{\dbr{\kappa +1}}$ with subsequently renaming this parameter into $\left[{\lambda}_i\right]^{\dbr{\kappa +1}}$ completes the proof of Theorem~\ref{thm:MLF_encl}.
\end{proof}

\begin{remark}
	Quasi-linear state-space representations of fractional-order differential equations in diagonally dominant form according to Def.~\ref{def:diag_dominant} can be simulated efficiently by the symbolically simplified iteration scheme
	\begin{equation}\label{eq:MLF_mod}
	\left[{\lambda}_i\right]^{\dbr{\kappa +1}}
	:= a_{ii}\rb{\sqb{{\mathbf{z}}_{e}}^{\dbr{\kappa}}\left(\sqb{t}\right)}
	+ \sum\limits_{
		{j=1}\atop{j \neq i}
	}^n\left\{
	a_{ij}\rb{\sqb{{\mathbf{z}}_{e}}^{\dbr{\kappa}}\left(\sqb{t}\right)} \cdot
	\frac{E_{\nu,1}\rb{\left[\lambda_j\right]^{\dbr{\kappa}}\cdot \sqb{t}^\nu}}
	{E_{\nu,1}\rb{\left[\lambda_i\right]^{\dbr{\kappa}}\cdot \sqb{t}^\nu}}
	\cdot
	\frac{
		\sqb{{{z}}_{e,j}}\left(0\right)
	}
	{
		\sqb{{{z}}_{e,i}}\left(0\right)
	}
	\right\} \enspace.
	\end{equation}
	In~(\ref{eq:MLF_mod}), the quotient of two Mittag-Leffler functions cannot be simplified further in the general case. This imposes further restrictions on the numerical evaluation of this iteration scheme by means of interval techniques as shown in the following section.
\end{remark}

\begin{remark}
	For the order $\nu=1$, the iteration formula in Theorem~\ref{thm:MLF_encl} becomes identical to the solution published, for example, in~\cite{Rauh_SCAN14} due to the identity $E_{1,1}(z) \equiv e^z$.
\end{remark}

\section{Evaluating Mittag-Leffler Functions for Interval Arguments}\label{sec:interval_MLF}
As described in~\cite{Haubold}, it is possible to determine rough outer bounds for the range of a Mittag-Leffler function that is evaluated over a finitely large real-valued interval argument. These bounds result from the fact that the Mittag-Leffler function can be interpreted as a continuous interpolation between Gaussian (exponential) and Lorentzian (rational) functions. Due to this property, the bounds 
\begin{equation}\label{eq:bound_ml}
\exp\rb{-\zeta} < E_{\nu,1}\rb{-\zeta} \le \frac{1}{1+\zeta} 
\qquad \text{and} \qquad
\exp\rb{-\zeta^2} < E_{\nu,1}\rb{-\zeta^2} \le \frac{1}{1+\zeta^2}
\end{equation}
can be determined for $\zeta \ge 0$. A visualization of these bounds in terms of the dependency on the parameter range $\nu \in \intv{0}{1}$ for the fractional differentiation order can be found in~\cite{Rauh_Acta_2020}.
%
%

Because these bounds are usually too conservative for the interval-based evaluation of the iteration procedure presented in the previous section, floating point evaluations of the Mittag-Leffler function using the \textsc{Matlab} implementation of R.~Garrappa~\cite{Garrappa_2015} are extended in the following to obtain tighter guaranteed interval bounds for the case of real-valued arguments to which this paper is restricted.

\subsection{Interval Evaluation of the Mittag-Leffler Function with Real Arguments}
\begin{theorem}[\cite{Rauh_ECC_2020,Rauh_Acta_2020}~\label{thm:MLF_int}Interval bounds for the Mittag-Leffler  function with real arguments]\mbox{$~~$}\newline
	Interval evaluations of the two-parameter Mittag-Leffler function with real-valued arguments \linebreak$z \in \sqb{z} = \intv{\ul{z}}{\ol{z}}$ are given by
	\begin{equation}
	\begin{split}
	E_{\nu,\beta}\rb{\sqb{z}} & \in \sqb{{E}^*_{\nu,\beta}\rb{\sqb{z}}} = 
	\sqb{\tilde{E}_{\nu,\beta}\rb{\sqb{z}}} + 
	\frac{\epsilon}{1+  \epsilon} \cdot 
	\rb{1+\left|\sqb{\tilde{E}_{\nu,\beta}\rb{\sqb{z}}}\right|} \cdot
	\intv{-1}{1}
	\end{split}
	\raisetag{13mm}
	\end{equation}
	with the tolerance value $\epsilon>0$ of a floating point evaluation of~(\ref{eq:Mittag_Leffler}) and the interval definition
	\begin{equation}\label{eq:MLF_int}
	\sqb{\tilde{E}_{\nu,\beta}\rb{\sqb{z}}} = \intv{\bigtriangledown \tilde{E}_{\nu,\beta}\rb{\ul{z}}}{\bigtriangleup \tilde{E}_{\nu,\beta}\rb{\ol{z}}}
	\end{equation}
	in which $\bigtriangledown$ and $\bigtriangleup$ denote switchings of the rounding mode toward minus and plus infinity, respectively, in a (CPU-based) floating point evaluation.
\end{theorem}


\subsection{Exploitation of Monotonicity in Interval Evaluations of the Mittag-Leffler Function}
To reduce overestimation in the interval evaluation of the iteration procedure according to Theorem~\ref{thm:MLF_encl}, monotonicity properties of the Mittag-Leffer function need to be accounted for with respect to the time $t$, the solution parameters $\lambda_i$, and the derivative order $\nu$. For a description of all corresponding monotonicity properties, the reader is referred to~\cite{Rauh_Acta_2020} and~\cite{Rauh_ECC_2020}.

\subsection{Bounding the Temporal Truncation Error due to the Infinite Memory Property of Fractional-Order Systems}\label{sec:discr_errors}
Fractional-order systems are characterized by an infinite memory of previous states~\cite{Podlubny,Oustaloup}. Therefore, the restart of the simulation of a dynamic system according to Def.~\ref{def:diag_dominant} at some point of time $t=t_{k+1}$ does not only have to account for state information $\mathbf{z}(t_{k+1})$ computed by a simulation originally initiated at some previous point $t_k < t_{k+1}$, but also the effect of a temporal truncation error that can be expressed by component-wise bounds originating from a replacement of the fractional-order derivative of order $\nu$ with a memory start at $t=t_k$ by a new starting point $t=t_{k+1}$. The corresponding derivative operators are subsequently denoted by ${}_{t_k}\mathcal{D}_t^\nu\mathbf{z}(t)$ and ${}_{t_{k+1}}\mathcal{D}_t^\nu\mathbf{z}(t)$, respectively. Note, the following Theorems~\ref{thm:mu} and~\ref{thm:contractor} are generalizations of those given in~\cite{Rauh_ECC_2020,Rauh_Acta_2020}, focusing especially on non-scalar applications.

\begin{theorem}[\label{thm:mu}Bounds for temporal truncation errors]
	Resetting the initial point of time of the integration of fractional-order models defined in Def.~\ref{def:diag_dominant} based on Theorem~\ref{thm:MLF_encl} after completion of a time interval of length $T$ requires the inflation of the right-hand side of the state equations by the symmetric interval $\intv{-\boldsymbol{\mu}}{\boldsymbol{\mu}}$ at the point $t_k+T$ with 
	\begin{equation}
	\boldsymbol{\mu} :=
	\frac{{\boldsymbol{\mathbcal{Z}}} \cdot  \rb{t_k+T}^{-\nu}}{\left|\Gamma(1-\nu)\right|}
	\end{equation}
	and the component-wise defined supremum of the set of reachable states
	\begin{equation}\label{eq:sup_States}
	\mathcal{Z}_i = \sup\limits_{t \in \intv{t_0}{t_{k+1}}} \left|z_i(t)\right|
	\enspace .
	\end{equation}
	The re-initialized initial value problem is then given by
	\begin{equation}\label{eq:new_ivp}
	{}_{t_k+T}\mathcal{D}_t^\nu\mathbf{z}(t) = 
	\mathbf{z}^{\rb{\nu}}(t) = 
	\mathbf{f}\rb{\mathbf{z}(t)} + \intv{-\boldsymbol{\mu}}{\boldsymbol{\mu}}
	=:\tilde{\mathbf{f}}\rb{\mathbf{z}(t)}
	\end{equation}
	with the initial state enclosure ${\mathbf{z}}(t_{k}+T) \in \sqb{\mathbf{z}}(t_{k}+T)$ resulting form the solution of 
	\begin{equation}
	{}_{t_0}\mathcal{D}_t^\nu\mathbf{z}(t) = 
	\mathbf{z}^{\rb{\nu}}(t) = 
	\mathbf{f}\rb{\mathbf{z}(t)}
	\quad 
	\text{with} 
	\quad 
	{\mathbf{z}}(t_{0}) \in \sqb{\mathbf{z}}(t_{0})~,~~t_0=0
	\enspace .
	\end{equation}
\end{theorem}

\begin{proof}
	Theorem~\ref{thm:mu} is a consequence of the 
	component-wise defined error bounds for a general fractional derivative operator of a commensurate system model on the time interval $t_k+T \le t \le t_{k+1}$ that can be computed according to
	\begin{equation}\label{eq:bounds_truncation}
	\left|{}_{t_k}\mathcal{D}_t^\nu\mathbf{z}(t) - {}_{t_k+T}\mathcal{D}_t^\nu\mathbf{z}(t)\right|
	\le
	\frac{{\mathbcal{Z}} T^{-\nu}}{\left|\Gamma(1-\nu)\right|}
	=: \boldsymbol{\mu}
	\end{equation}
	In correspondence with~\cite{Podlubny}, Eq.~(\ref{eq:bounds_truncation}) relies on the component-wise defined supremum~(\ref{eq:sup_States}) of the reachable states denoted by the vector ${\mathbcal{Z}}$.
\end{proof}

\begin{theorem}[\label{thm:contractor}Contractor for the state enclosure of fractional-order systems]
	Assume that a reference solution $\mathbf{z}(t) \in \sqb{\mathbf{z}_{\mathrm{ref}}}(t)$ has already been computed for the initial value problem with the initial point of time $t=t_k$ that is valid up to the point $t = t^* > t_k+T$ 
	and that the application of Theorem~\ref{thm:MLF_encl} to the re-initialized initial value problem~(\ref{eq:new_ivp}) in Theorem~\ref{thm:mu}
	with the initial point of time $t=t_k+T$ has provided the interval bounds
	$\sqb{\mathbf{z}_e}(t) = 
	\exp\left(\left[\boldsymbol\Lambda\right]\cdot \rb{t-(t_k+T)}\right) \cdot
	\sqb{{\mathbf{z}}_{e}}\left(t_k+T\right)$ that are also valid up to $t = t^*$ 
	with the associated solution parameters $\sqb{\lambda_i}$, $i \in \{1,\ldots,n\}$. A contractor is then given by
	\begin{equation}
	\sqb{\lambda_i} := \sqb{\lambda_i} \cap \sqbs{\tilde{\lambda}_i}
	\end{equation}
	with
	\begin{equation}\label{eq:contr_lambda}
	\sqbs{\tilde{\lambda}_i} = \frac{\tilde{f}_i\rb{\sqb{\mathbf{z}_e}(\intv{t_k+T}{t^*})}
		\cap \tilde{f}_i\rb{\sqb{\mathbf{z}_{\mathrm{ref}}}(\intv{t_k+T}{t^*})}
	}{
		{\sqb{z_{e,i}}(\intv{t_k+T}{t^*})} \cap {\sqb{z_{i,\mathrm{ref}}}(\intv{t_k+T}{t^*})}
	}
	\enspace .
	\end{equation}
\end{theorem}

\begin{proof}
	The validity of Theorem~\ref{thm:contractor} is a direct consequence of the fact that both $\sqb{\mathbf{z}_{\mathrm{ref}}}(t)$ and $\sqb{\mathbf{z}_e}(t)$ are verified state enclosures according to 
	$\mathbf{z}(t) \in \sqb{\mathbf{z}_{\mathrm{ref}}}(t)$ and $\mathbf{z}(t) \in \sqb{\mathbf{z}_e}(t)$ and thus have to satisfy the fractional-order differential equation in the componentwise notation~(\ref{eq:exp1_fde}). Intersecting the evaluation of~(\ref{eq:exp1_fde}) for both state enclosures and solving the result for the interval of the solution parameter $\sqbs{\tilde{\lambda}_i}$ yields the relation~(\ref{eq:contr_lambda}).
\end{proof}

\section{Numerical Examples}\label{sec:example}

\subsection{Interval Bounds for Different Integration Horizons for a Nonlinear Example}\label{sec:ex_nonlin}
Consider the nonlinear example 
\begin{equation}\label{eq:example}
z^{(\nu)}(t) = p\cdot z^3(t) = p \cdot a(z(t)) \cdot z(t)
\end{equation}
with an uncertain initial state $z(0) \in \sqb{z}(0)$, the interval parameter $p \in \sqb{p}$, and the uncertain differentiation order $\nu \in \sqb{\nu}$ 
for the application of Theorem~\ref{thm:MLF_encl}. This example was already studied in~\cite{Rauh_ECC_2020} without a subdivision of the integration time horizon as well as without a restart of the integration according to Sec.~\ref{sec:discr_errors}. Due to the quasi-linear formulation of Eq.~(\ref{eq:example}), the modified iteration formula given in~(\ref{eq:MLF_mod}) becomes directly applicable.\\

\noindent
In the following, two cases with different amount of uncertainty are distinguished:
\begin{description}
	\item[Case~a:] $\sqb{z}(0) = \intv{0.99}{1.0}$,
	$\sqb{p}  = \intv{-2}{-1.99}$,
	$\sqb{\nu} = \intv{0.8}{0.81}$ and
	\item[Case~b:] $\sqb{z}(0) = \intv{0.5}{1.0}$,
	$\sqb{p}  = \intv{-2}{-1}$,
	$\sqb{\nu} = \intv{0.8}{0.9}$.
\end{description}

Fig.~\ref{fig:example} provides a comparison of the effect of reducing the integration step size $T$ according to the Theorems~\ref{thm:mu} and~\ref{thm:contractor}. In all presented results, an equidistant discretization step size was applied. Note, optimizing either an equidistant or adaptive integration step size will be one of the subjects of future work.
It can be seen that --- especially for \textbf{Case~a} --- the reduction of the integration step size leads to a significant reduction of overestimation in the computed state enclosures. This effect is not that much visible for \textbf{Case~b}, where step sizes smaller than $T = 2^{-2}$ lead to almost identical results, except for the very beginning of the integration time horizon. In both cases, it can be observed that the lower interval bounds of the state variable $z(t)$ are seemingly not influenced by overestimation. This can be concluded from the fact that a reduction of overestimation only occurs at the computed supremum when reducing the step size $T$.

\begin{figure}[htp]
	\centering
	\subfloat[{\label{fig_small1}}State enclosures for \textbf{Case~a}; integration step sizes \mbox{$T=1$} (blue),  $T=2^{-1}$ (red), $T=2^{-2}$ (yellow).]{\resizebox{0.45 \linewidth}{!}{{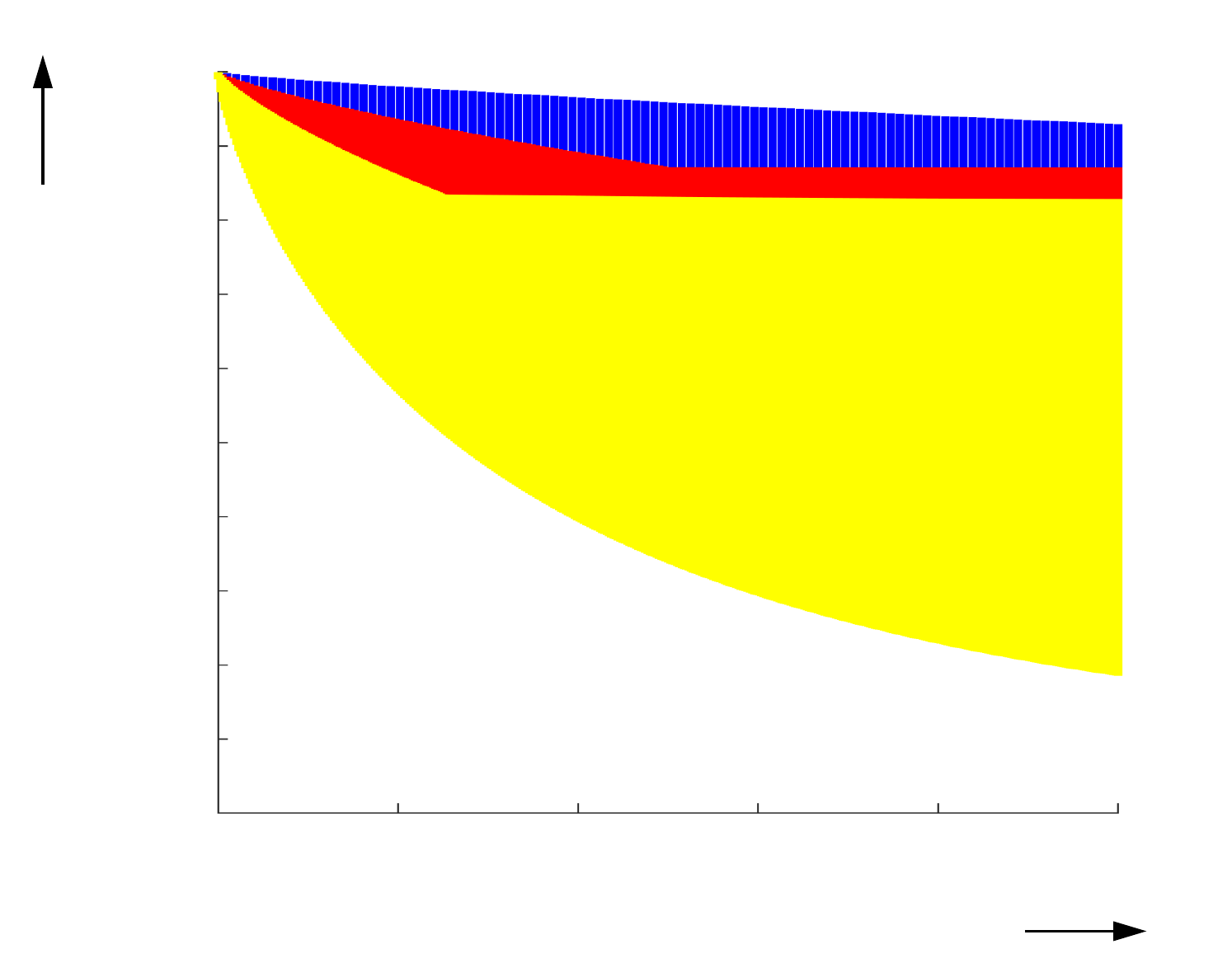}}}
	\hfill
	\subfloat[{\label{fig_large1}}State enclosures for \textbf{Case~b}; integration step sizes \mbox{$T=1$} (blue),  $T=2^{-1}$ (red), $T=2^{-2}$ (yellow).]{\resizebox{0.45 \linewidth}{!}{{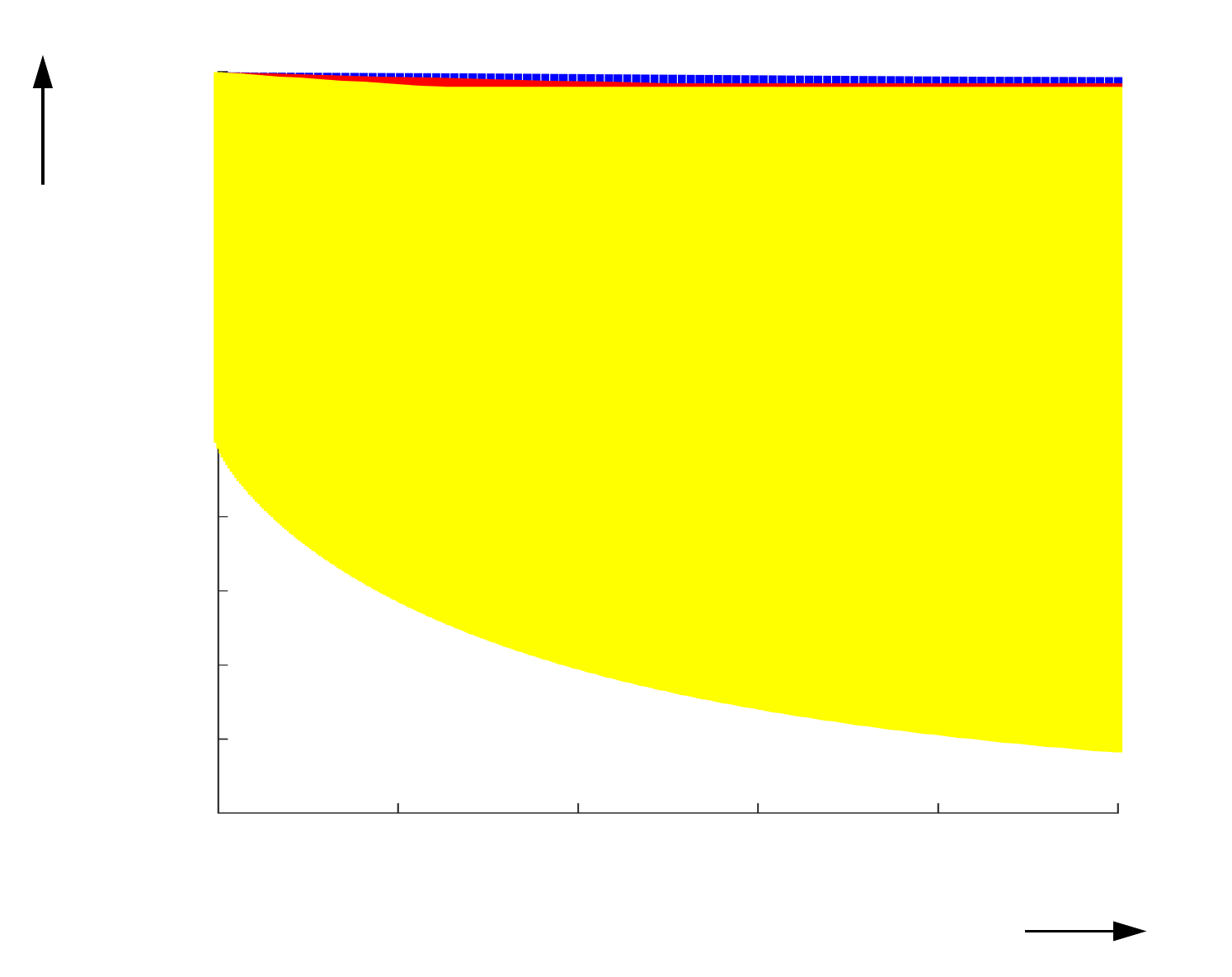}}}
	\\
	\subfloat[{\label{fig_small2}}State enclosures for \textbf{Case~a}; integration step sizes \mbox{$T=2^{-2}$} (blue),  $T=2^{-3}$ (red), $T=2^{-4}$ (yellow).]{\resizebox{0.45 \linewidth}{!}{{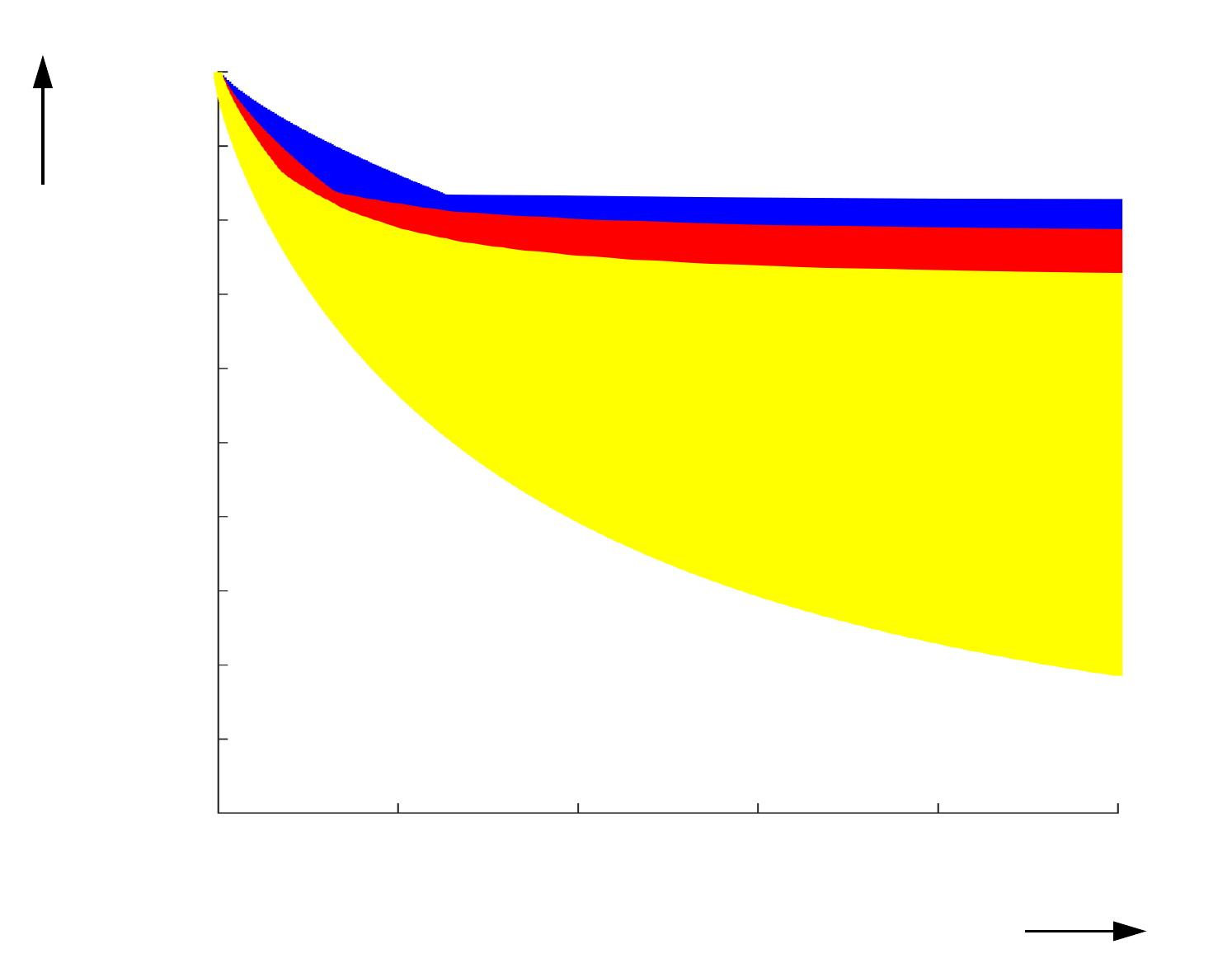}}}
	\hfill
	\subfloat[{\label{fig_large2}}State enclosures for \textbf{Case~b}; integration step sizes \mbox{$T=2^{-2}$} (blue),  $T=2^{-3}$ (red), $T=2^{-4}$ (yellow).]{\resizebox{0.45 \linewidth}{!}{{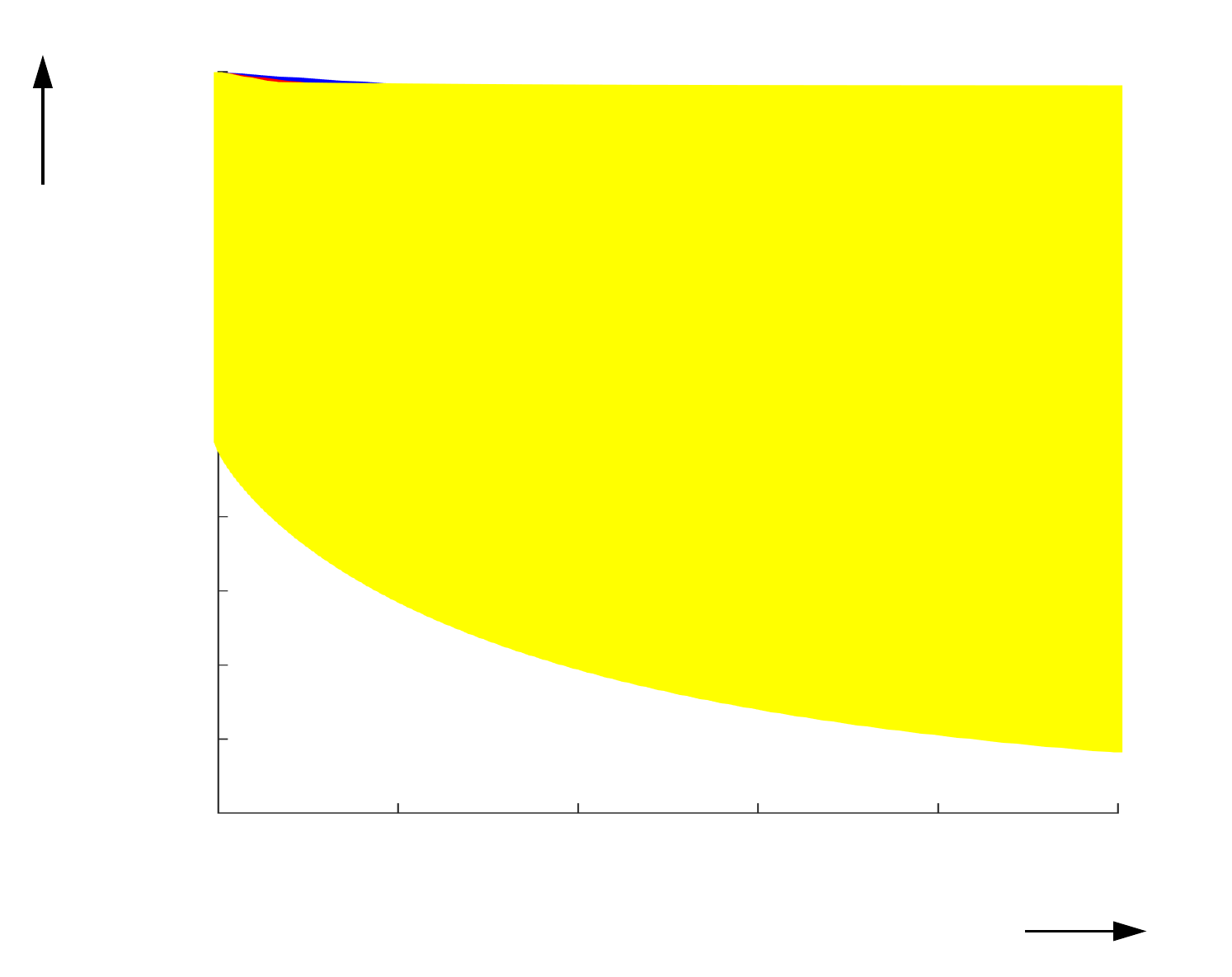}}}
	\\
	\subfloat[{\label{fig_small3}}State enclosures for \textbf{Case~a}; integration step sizes \mbox{$T=2^{-2}$} (blue),  $T=2^{-4}$ (red), $T=2^{-5}$ (yellow).]{\resizebox{0.45 \linewidth}{!}{{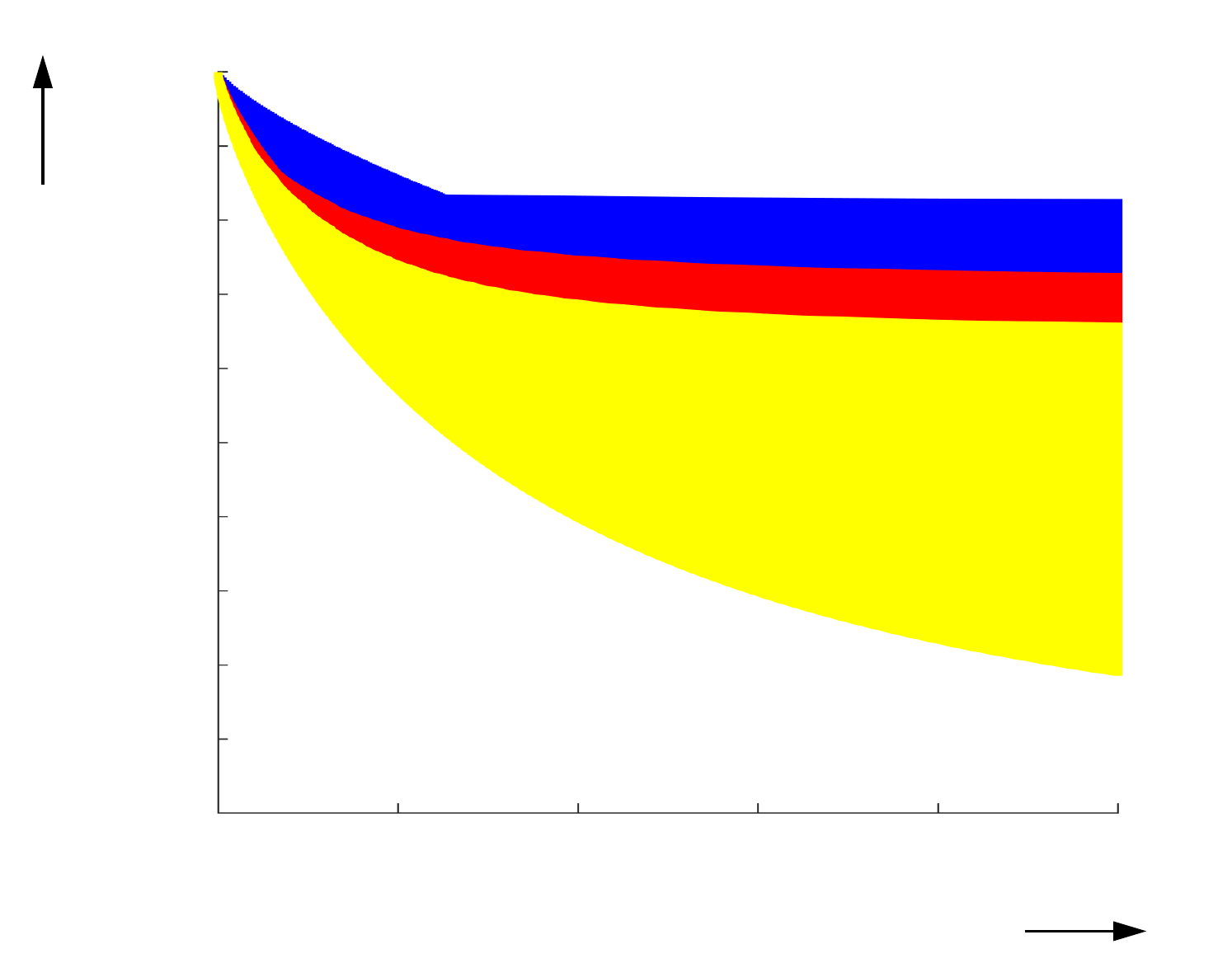}}}
	\hfill
	\subfloat[{\label{fig_large3}}State enclosures for \textbf{Case~b}; integration step sizes \mbox{$T=2^{-2}$} (blue),  $T=2^{-4}$ (red), $T=2^{-5}$ (yellow).]{\resizebox{0.45 \linewidth}{!}{{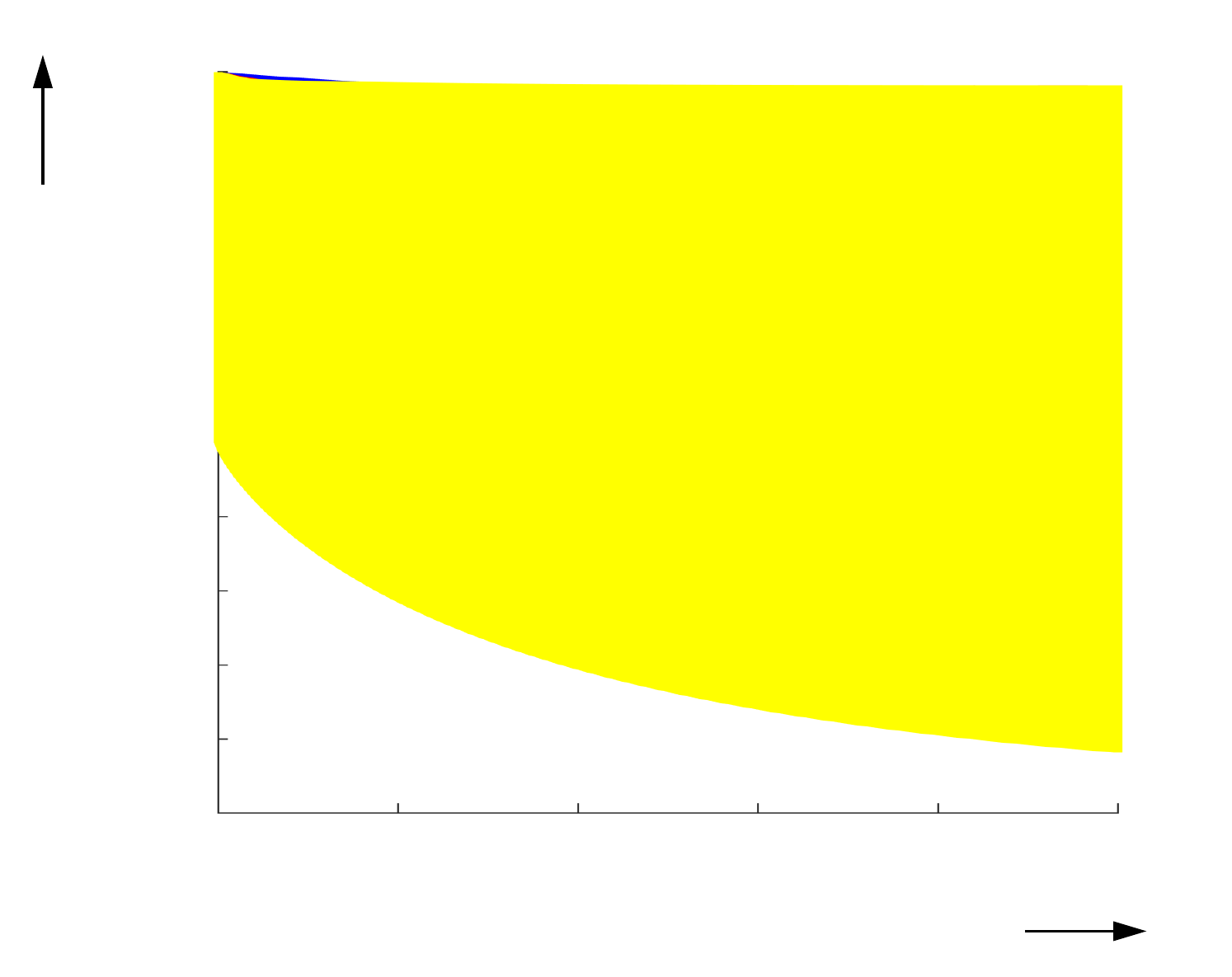}}}
	\\
	\caption{{\label{fig:example}}Guaranteed state enclosures for different discretization step sizes and levels of uncertainty. }
\end{figure}

In addition, it should be pointed out that all presented state enclosures were determined under the assumption that the parameters $\sqb{p}$ and $\sqb{\nu}$ can vary arbitrarily within their respective interval bounds over the complete integration time horizon $t \in \intv{0}{1}$. As soon as more precise knowledge about the worst-case parameter variation rates becomes available, for example, if physical reasoning justifies that these values can be set to unknown but constant quantities, the state enclosures can also be tightened in \textbf{Case~b}. Then, it becomes possible to subdivide the ranges of these parameter values and to perform individual simulations for respective subdomains, where the sets of reachable states are then characterized by the convex outer hull over all individual solution tubes.

\subsection{Simplified Fractional-Order Battery Model}
As a second application scenario, consider the fractional-order battery model depicted in Fig.~\ref{fig:battery}. It describes the dynamics of the charging and discharging behavior of a Lithium-Ion battery with the help of the state of charge $\sigma(t)$ as well as with the dynamics of the exchange of charge carriers in the interior of the battery cell which is related to electrochemical double layer effects. The latter is represented with the help of the voltage drop $v_{1}(t)$ across the fractional-order constant phase element $Q$ which serves as a generalization of an ideal capacitor according to~\cite{Rauh_IFAC_2020,ANDRE20115334,ZOU2018286}. This kind of constant phase element was already motivated by the example in Sec.~\ref{sec:lin_frac_comparison}.

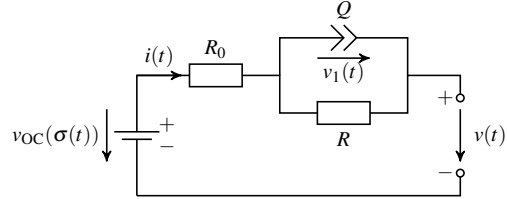
\begin{figure}[htp]
	\centering
	\begin{tikzpicture}[font=\scriptsize,scale=1] 
	\draw[semithick,white] (0,0) to (8.5cm,0);  
	\tikzmath{\rx=+0.70cm; \rh=+0.30cm;};
	\tikzmath{\qx=+0.17cm; \qy=+0.34cm; \qz=+0.17cm;};
	\tikzmath{\linx=+0.50cm; \liny=+0.50cm;};
	\tikzmath{\tex=0.20cm;};	
	\tikzmath{\arrl=+0.07cm;};	
	
	\node (vOC1) at (+0.60cm,+0.00cm) {$v_\mathrm{OC}(\sigma(t))$~~~~~~}; 
	\ExtractCoordinate{vOC1} 	\draw[semithick,->] (\xCoord+0.5cm,+0.35cm) to (\xCoord+0.5cm,-0.35cm) coordinate (vOC2);
	\ExtractCoordinate{vOC2} 	\draw[semithick](\xCoord+0.2cm, -0.05cm) to (\xCoord+0.6cm, -0.05cm) coordinate (vOC31);
	\ExtractCoordinate{vOC2} 	\draw[semithick](\xCoord+0.1cm,+0.05cm) to (\xCoord+0.7cm,+0.05cm) coordinate (vOC32);
	\ExtractCoordinate{vOC31}	\draw[semithick](\xCoord -0.2cm,\yCoord) to (\xCoord-0.2cm,\yCoord -0.75cm) coordinate(vOC4);
	\ExtractCoordinate{vOC32}	\draw[semithick](\xCoord -0.3cm,\yCoord) to (\xCoord-0.3cm,\yCoord+0.75cm) to (\xCoord+0.4cm,\yCoord+0.75cm) coordinate (R01);
	\node at (\xCoord+0.1cm,+0.15cm) {$+$};
	\node at (\xCoord+0.1cm, -0.15cm) {$-$};
	\ExtractCoordinate{R01}	\draw[semithick,->] 	(\xCoord-0.7cm,\yCoord) -- (\xCoord-0.1cm,\yCoord); 
	\node at (\xCoord -0.4cm,\yCoord+0.25cm) {$i(t)$};
	\ExtractCoordinate{R01} 	\draw[semithick]	( \xCoord					, \yCoord				) -- ( \xCoord				, \yCoord+\rh/2	) -- 
	( \xCoord+\rx			, \yCoord+\rh/2	) -- ( \xCoord+\rx		, \yCoord				) coordinate (R02) -- 
	( \xCoord+\rx			, \yCoord -\rh/2	) -- ( \xCoord				, \yCoord -\rh/2	) -- 
	( \xCoord					, \yCoord				) -- cycle; 
	\node at (\xCoord+\rx/2, \yCoord+\rh/2+\tex) {$R_{0}$};
	\ExtractCoordinate{R02} 	\draw[semithick]	( \xCoord					, \yCoord				) -- (\xCoord+\linx		, \yCoord				) -- 
	( \xCoord+\linx		, \yCoord -\liny	) -- (\xCoord+2*\linx	, \yCoord -\liny	) coordinate (R11);
	\ExtractCoordinate{R11} 	\draw[semithick]	( \xCoord					, \yCoord				) -- ( \xCoord				, \yCoord+\rh/2	) -- 
	( \xCoord+\rx			, \yCoord+\rh/2	) -- ( \xCoord+\rx		, \yCoord				) coordinate (R12) -- 
	( \xCoord+\rx			, \yCoord -\rh/2	) -- ( \xCoord				, \yCoord -\rh/2	) -- 
	( \xCoord					, \yCoord				) -- cycle; 
	\node at (\xCoord+\rx/2, \yCoord-\rh/2-\tex) {$R$};
	\ExtractCoordinate{R12} 	\draw[semithick]	( \xCoord					, \yCoord				) -- ( \xCoord+\linx		, \yCoord				) -- 
	( \xCoord+\linx		, \yCoord+\liny	) coordinate(R13) -- ( \xCoord+\linx+0.7cm	, \yCoord+\liny	) --
	( \xCoord+\linx+0.7cm	, \yCoord+\liny-0.30cm	) coordinate (R21);
	\ExtractCoordinate{R21} 	\draw[semithick,<-] (\xCoord+0.00cm, -0.35cm) -- (\xCoord+0.00cm,+0.35cm); \node at (\xCoord+0.4cm,+0.00cm) {$v(t)$};
	\draw[semithick,fill=white] (\xCoord,\yCoord) circle (0.05cm) coordinate(C1);
	\draw[semithick] ( \xCoord, \yCoord -1.00cm) -- ( \xCoord, \yCoord -1.30cm) -- (vOC4);														
	\draw[semithick,fill=white] (\xCoord,\yCoord -1.00cm) circle (0.05cm) coordinate(C2);
	\ExtractCoordinate{C1} 		\node at (\xCoord-0.20cm,\yCoord+0.00cm) {$+$};
	\ExtractCoordinate{C2} 		\node at (\xCoord-0.20cm,\yCoord -0.00cm) {$-$};
	\ExtractCoordinate{R02} 	\draw[semithick]	( \xCoord+\linx		, \yCoord				) -- ( \xCoord+\linx					, \yCoord+\liny	) -- 
	( \xCoord+\linx		, \yCoord+\liny	) -- ( \xCoord+2*\linx+\rx/2	, \yCoord+\liny	) coordinate (Q11);
	\ExtractCoordinate{Q11} 	\draw[semithick]	( \xCoord-\qz			, \yCoord-\qy/2	) -- ( \xCoord		,\yCoord) -- ( \xCoord-\qz, \yCoord+\qy/2);
	\draw[semithick]	( \xCoord-\qz+\qx	, \yCoord-\qy/2	) -- ( \xCoord+\qx,\yCoord) coordinate(Q12) -- ( \xCoord-\qz+\qx, \yCoord+\qy/2);
	\node at (\xCoord, \yCoord+\qy/2+\tex) {$Q$};
	\draw[semithick,->] ( \xCoord-0.35cm, \yCoord-\qy/2-0.10cm) -- ( \xCoord+0.35cm, \yCoord-\qy/2-0.10cm);
	\node at (\xCoord, \yCoord-\qy/2-0.3cm) {$v_{1}(t)$};														
	\ExtractCoordinate{Q12} 	\draw[semithick]	( \xCoord					, \yCoord				) -- ( \xCoord+\linx+\rx/2-\qx	, \yCoord	) -- 
	( \xCoord+\linx+\rx/2-\qx, \yCoord-\liny ) ;																											
	\end{tikzpicture}\vspace{-1mm}
	\caption{\label{fig:battery}Basic fractional-order equivalent circuit model of batteries.}
\end{figure}

For the following investigation of the proposed interval-based simulation approach for fractional-order system models, assume that the state vector
\begin{equation}
\mathbf{x}(t) = \begin{bmatrix} \sigma(t) & {}_0\fod_{t}^{0.5}\sigma(t) & v_{1}(t) \end{bmatrix}^{T}
\in \mathbb{R}^3
\end{equation}
comprises in addition to the voltage $v_1$ across the constant phase element in Fig.~\ref{fig:battery} and the state of charge also its fractional derivative of order $\nu= 0.5$. Then, following the modeling steps described in~\cite{Rauh_IFAC_2020}, where this dynamic system representation was employed for the derivation of a cooperativity-enforcing interval observer design, and generalizing the charging/discharging dynamics to 
\begin{equation}
{}_0\fod_{t}^{1}\sigma(t) = -\frac{\eta_0 \cdot i(t) + \eta_1 \cdot \sigma(t) \cdot \mathrm{sign}\rb{i(t)}}{3600 C_\mathrm{N}}
\end{equation}
 with the terminal current $i(t)$, the commensurate-order quasi-linear state equations
\begin{equation}\label{eq:state_battery} 
{}_0\fod_{t}^{0.5} \mathbf{x}(t) = \mathbcal{A}\cdot \mathbf{x}(t) + \mathbcal{b}\cdot i(t)
\end{equation}
with the system and input matrices
\begin{equation}
\mathbcal{A} = \begin{bmatrix} 0 	& \phantom{-}1 & \phantom{-} 0 \\
\frac{\eta_1 \cdot \mathrm{sign}\rb{i(t)}}{3600 C_\mathrm{N}}  & \phantom{-}0 & \phantom{-} 0 \\
0  & \phantom{-}0 & -\frac{1}{RQ} \\ \end{bmatrix}
\quad \text{and} \quad
\mathbcal{b} = \begin{bmatrix} \phantom{-}0 \\ -\frac{\eta_0}{3600 C_\mathrm{N}} \\ \phantom{-}\frac{1}{Q} \end{bmatrix}
\end{equation}
and the terminal voltage 
			\begin{equation}
v(t) = \begin{bmatrix} \sum\limits_{k=0}^{4} c_{k}\sigma^{k-1}(t) && 0 && -1 \end{bmatrix}\cdot \mathbf{x}(t) + \rb{-R_{0} +d_{0}e^{d_{1}\sigma(t)}} \cdot i(t)
\end{equation}
as the system output are obtained. 
The following numerical simulations are based on the system parameters summarized in Tab.~\ref{tab:parameter} which were --- except for $\eta_1$ --- identified according to the experimental data referenced in~\cite{reuter2016,Rauh_IFAC_2020} during multiple charging and discharging cycles before the occurrence of aging. Now, a linear state feedback controller with the structure presented in~\cite{Kersten2019646} for the discharging phase ($i(t)>0$) is parameterized by assigning the asymptotically stable eigenvalues $\lambda \in \{-0.0001; -0.0002; -0.4832\}$ to the closed-loop dynamics, where the last value corresponds to the fastest asymptotically stable dynamics of the open-loop system.

\newcolumntype{Y}{>{\centering\arraybackslash}X}
\begin{table}[hb] \captionsetup{width=.5\textwidth} \caption{Parameters of the Lithium-Ion battery model.} \label{tab:parameter} \centering
	\begin{tabularx}{0.85\textwidth}{Y | Y | Y | Y | Y | Y | Y}
		\hline
		$R_{0}\,[\mathrm{\Omega}]$	& $Q\,[\mathrm{F/s^{0.5}}]$ 	& $R\,[\mathrm{\Omega}]$	& $\eta_0	\,[-]$	& $\eta_1	\,[-]$	& $C_\mathrm{N}\,[\mathrm{Ah}]$		& $c_{0}\,[\mathrm{V}]$	\\	
		\hline
	\mbox{$+1.7 \!\cdot\! 10^{-5}$}								& $+20.591$										& $+0.1005$									& $+1.0000$	& $+0.1000$ & $+3.1000$	 							& $+3.0607$	\\
		\hline
	\end{tabularx} 

	\begin{tabularx}{0.85\textwidth}{Y | Y | Y | Y | Y | Y}
		\hline
		$c_{1}\,[\mathrm{V}]$		& $c_{2}\,[\mathrm{V}]$		& $c_{3}	\,[\mathrm{V}]$ 	& $c_{4}\,[\mathrm{V}]$		& $d_{0}\,[\mathrm{\Omega}]$ 	& $d_{1}	\,[\mathrm{-}]	$\\
		\hline
		$+3.2965$							& $-8.3942\phantom{-}$		& $+11.088$							&  $-4.8992$							& $-0.2477$									&  $-14.302$ \\
		\hline
	\end{tabularx}
\end{table}

This leads to the following closed-loop system model
\begin{equation}\label{eq:state_battery_cl} 
\fod_{t}^{0.5} \mathbf{x}(t) = \rb{\mathbcal{A}-\mathbcal{b}\mathbcal{k}^T}\cdot \mathbf{x}(t) =: \mathbcal{A}_\mathrm{C} \cdot \mathbf{x}(t) \enspace ,
\end{equation}
where the entries $\mathcal{A}_{\mathrm{C},2,1}$, $\mathcal{A}_{\mathrm{C},2,2}$, $\mathcal{A}_{\mathrm{C},2,3}$ are inflated to interval parameters by symmetric bounds of $1\%$ radius of the respective nominal quantity and $\mathcal{A}_{\mathrm{C},3,3}$ to $10\%$, respectively. The matrix of eigenvectors at the interval midpoint corresponds to the transformation matrix $\mathbf{T}$ according to Def.~\ref{def:diag_dominant}. This allows for simulating the battery model both in transformed and initial coordinates $\mathbf{z}(t)$ and $\mathbf{x}(t)$, respectively, (cf.~Defs.~\ref{def:quasi_lin} and~\ref{def:diag_dominant}) and to intersect both enclosures after transforming the model in the form~(\ref{eq:diag_dom}) back into the original state space as soon as the simulation has been completed. Note, the system matrix  $\mathbf{A} \in \sqb{\mathbf{A}}$ in the uncertain diagonally dominant model is not Metzler, so that without the additional transformation procedure presented in~\cite{Kersten2019646}, it is necessary to apply the iteration from Theorem~\ref{thm:MLF_encl}.

In Fig.~\ref{fig:battery_sim}, simulations for the following two sets of initial conditions are compared. The left column of this figure corresponds to a small uncertainty level in the initial states according to
\begin{equation}
{\mathbf{x}}(0) \in 
\begin{bmatrix}
\intv{\phantom{-}0.99000}{\phantom{-}1.01000} \\
\intv{-0.00101}{-0.00099} \\
\intv{\phantom{-}0.09900}{\phantom{-}0.10100} 
\end{bmatrix}
\quad \text{resulting in} \quad
{\mathbf{z}}(0) \in 
\begin{bmatrix}
\intv{\phantom{-}0.089139}{\phantom{-}0.091348} \\
\intv{\phantom{-}7.880509}{\phantom{-}8.120525} \\
\intv{-9.110572}{-8.890557} 
\end{bmatrix}
\enspace, 
\end{equation}
while the right column accounts for the larger uncertainties
\begin{equation}
{\mathbf{x}}(0) \in 
\begin{bmatrix}
\intv{\phantom{-}0.9000}{\phantom{-}1.1000} \\
\intv{-0.0011}{-0.0009} \\
\intv{\phantom{-}0.0900}{\phantom{-}0.1100} 
\end{bmatrix}
\quad \text{with} \quad
{\mathbf{z}}(0) \in 
\begin{bmatrix}
\intv{\phantom{-}0.079204}{\phantom{-}0.101283} \\
\intv{\phantom{-}6.800442}{\phantom{-}9.200592} \\
\intv{-10.10064}{-7.900495} 
\end{bmatrix}
\enspace .
\end{equation}

\begin{figure}[htp]
	\centering
	\subfloat[{\label{fig_small_SOC}}State of charge $\sigma(t)$.]{\resizebox{0.45 \linewidth}{!}{{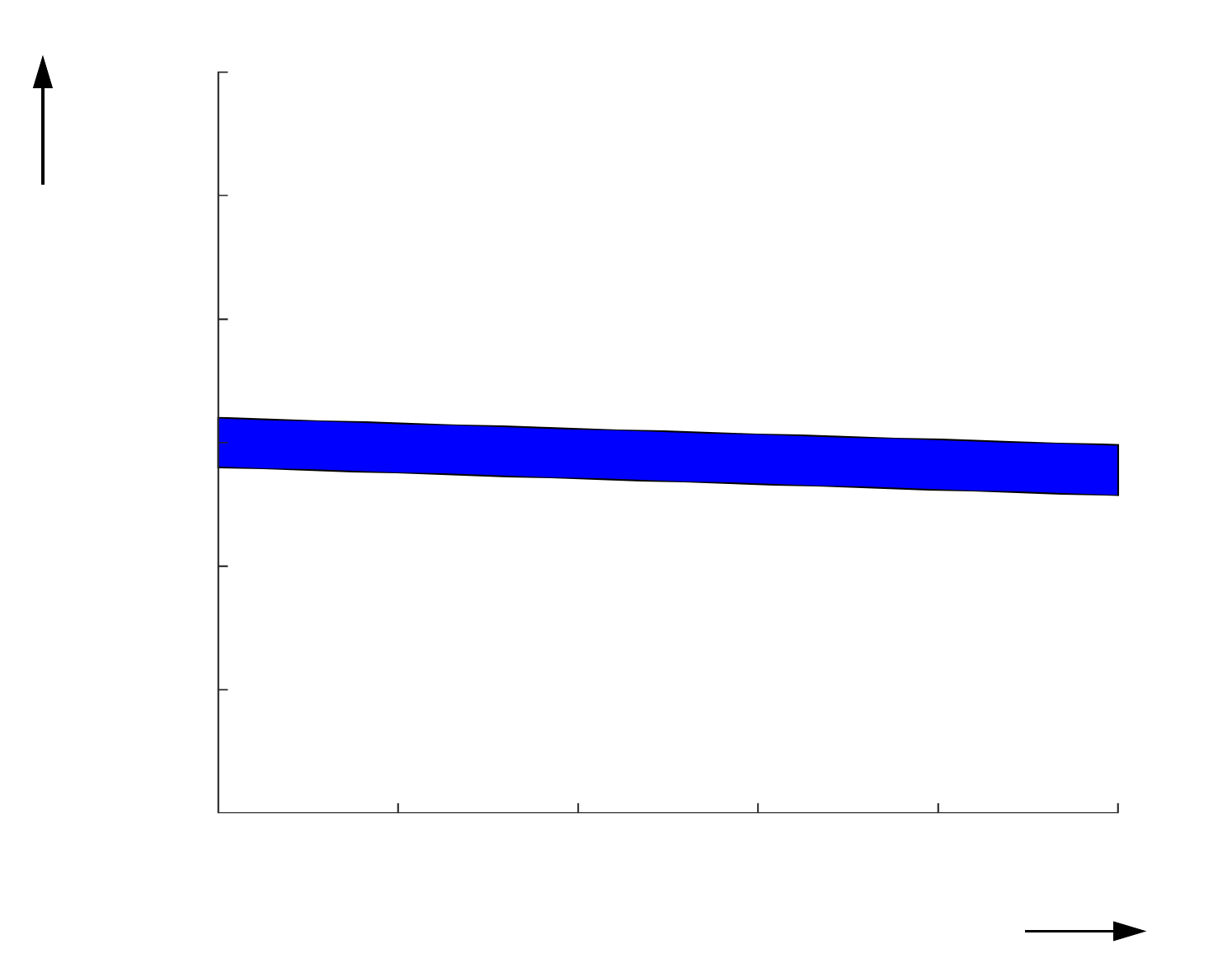}}}
	\hfill
	\subfloat[{\label{fig_large_SOC}}State of charge $\sigma(t)$.]{\resizebox{0.45 \linewidth}{!}{{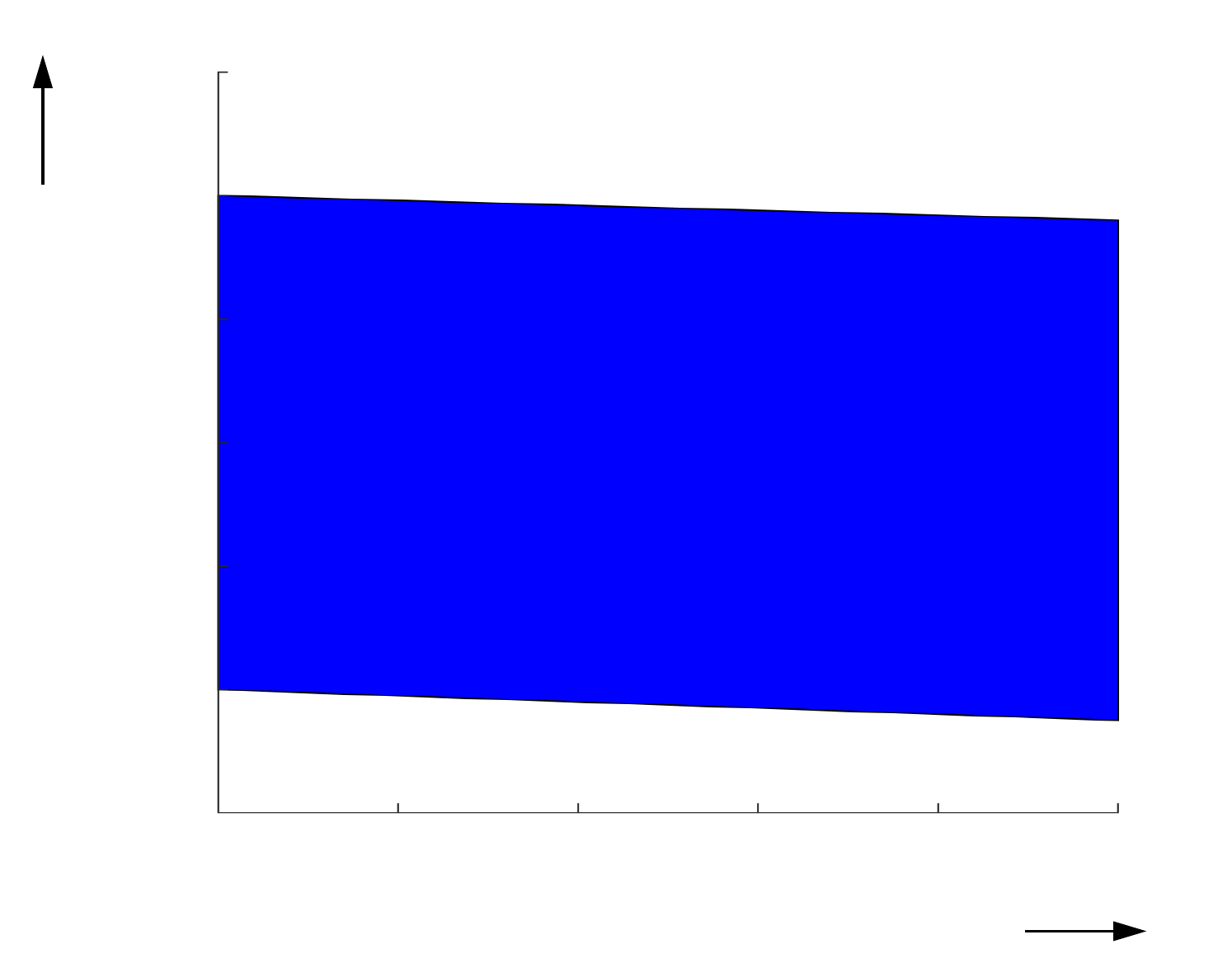}}}
	\\
	\subfloat[{\label{fig_small_CPE}}Voltage $v_1(t)$ across the constant phase element $Q$.]{\resizebox{0.45 \linewidth}{!}{{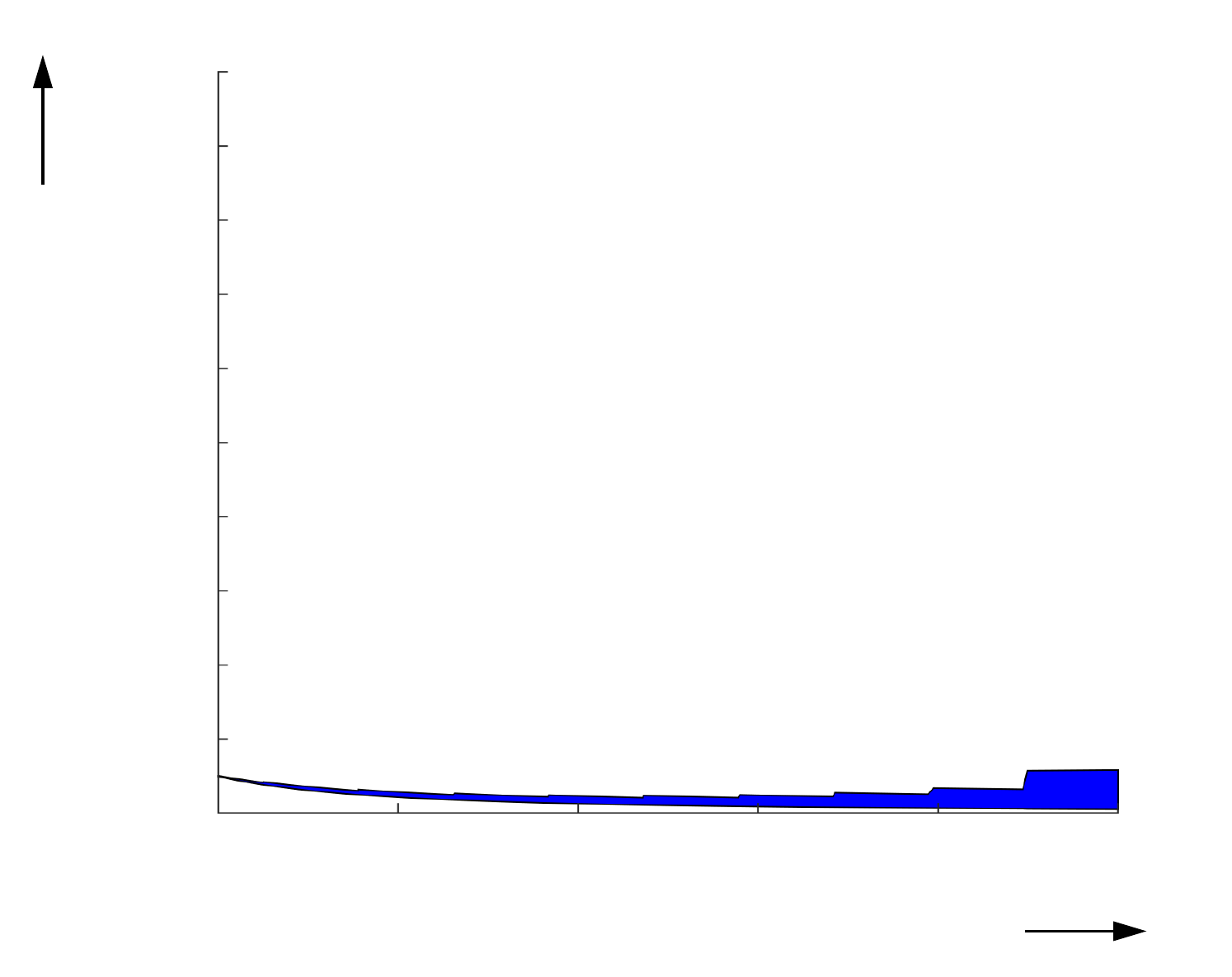}}}
	\hfill
	\subfloat[{\label{fig_large_CPE}}Voltage $v_1(t)$ across the constant phase element $Q$.]{\resizebox{0.45 \linewidth}{!}{{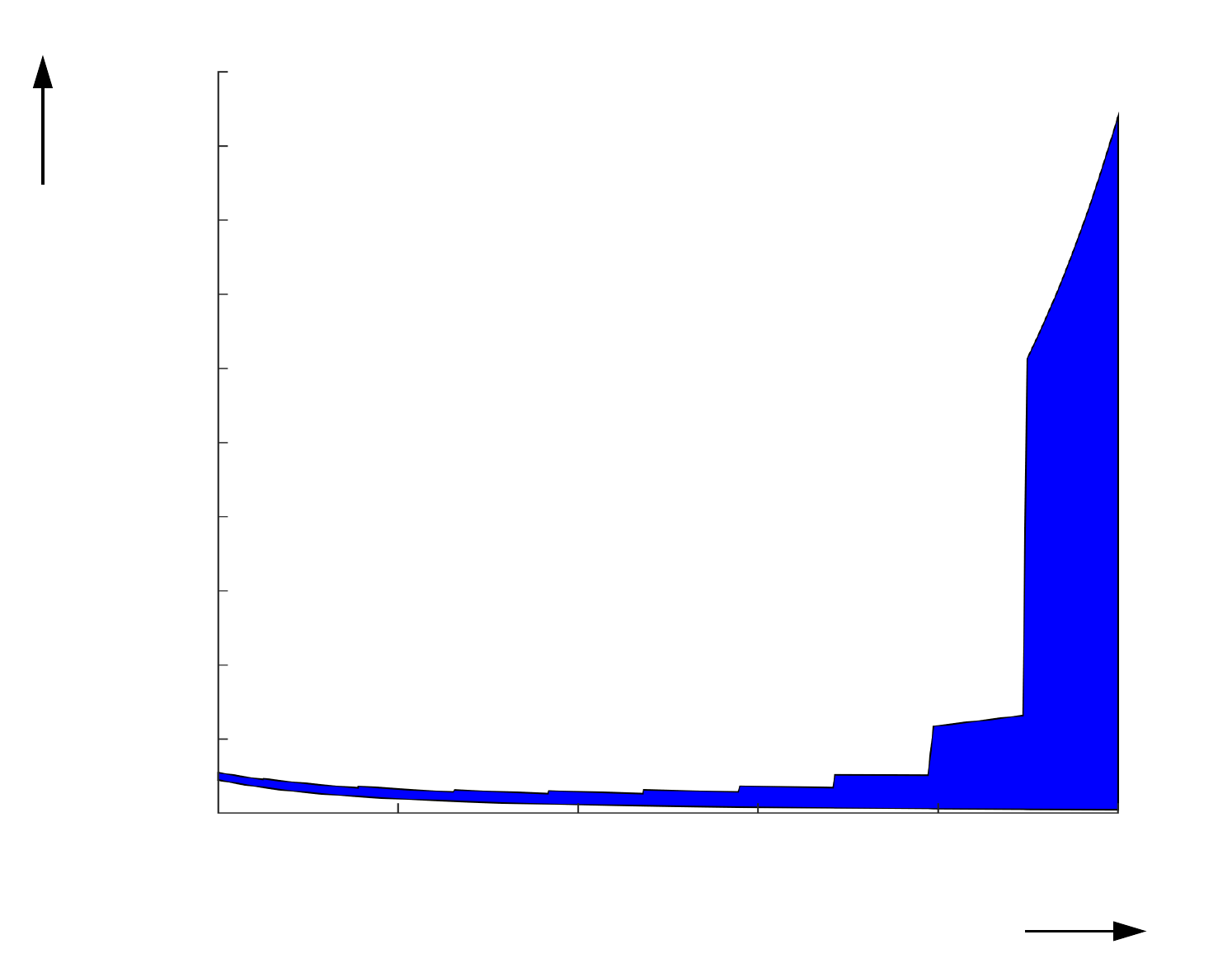}}}
	\\
	\subfloat[{\label{fig_small_out}}Battery terminal voltage $v(t)$.]{\resizebox{0.45 \linewidth}{!}{{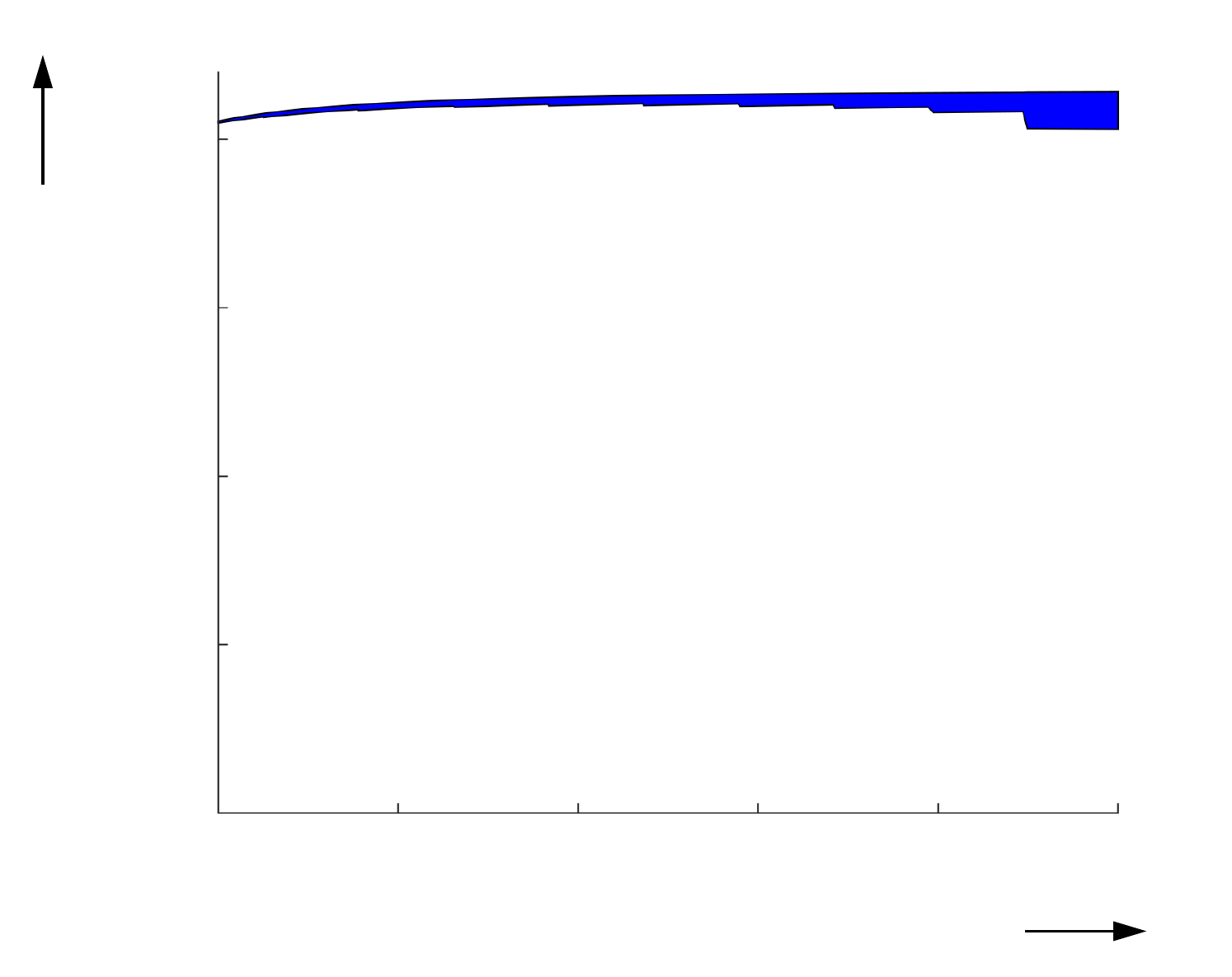}}}
	\hfill
	\subfloat[{\label{fig_large_out}}Battery terminal voltage $v(t)$.]{\resizebox{0.45 \linewidth}{!}{{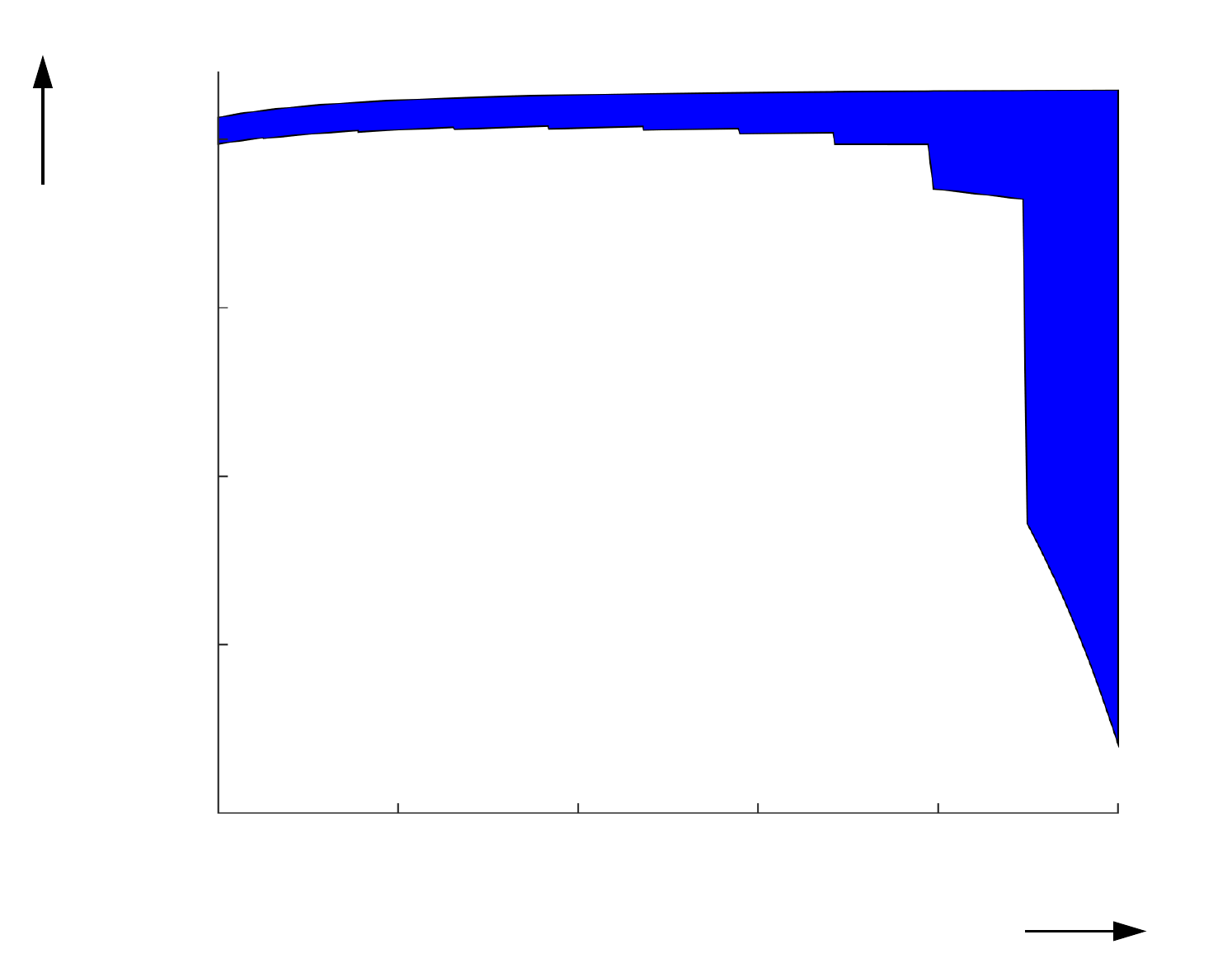}}}
	\\
	\caption{{\label{fig:battery_sim}}State and output voltage enclosures for the simplified fractional-order battery model in Fig.~\ref{fig:battery}, where the enclosures for the terminal voltage are determined with the help of an interval subdivision approach to reduce overestimation caused by the dependency effect.}
\end{figure}

Here, the simulation was performed without the subdivision of the time horizon according to Theorems~\ref{thm:mu} and~\ref{thm:contractor}. Instead, time-invariant parameter enclosures $\sqb{\lambda_i}$ were determined for 10~different time horizons from the range $t \in \intv{0}{10}\,\mathrm{s}$. The endpoints of these time intervals are visible in Fig.~\ref{fig:battery_sim} as those points at which the depicted enclosures widen in a step-like manner because parameters $\sqb{\lambda_i}$ that are valid for longer time spans  need to be wider to account for cross-coupling between multiple system states. 

In addition, Fig.~\ref{fig:battery_sim} visualizes that the use of tighter initial state intervals reduces the tendency of state enclosures to blow up over time if time-invariant bounds $\sqb{\lambda_i}$ are considered. Hence, future work for multi-dimensional systems which cannot be decoupled fully by the transformation according to Def.~\ref{def:diag_dominant}, needs to account for both an optimal step size control strategy with respective bounds for the temporal truncation errors and simultaneously bisectioning (or multisectioning) strategies of the initial state domains to counteract the wrapping effect especially in the initial phases of the simulation. There, stiff dynamics of the state equations may lead to a blow-up of the computed bounds if only a single transformation matrix $\mathbf{T}$ is employed for an approximate decoupling of the state equations. This, however, could be counteracted if multiple transformations are considered in parallel, where their corresponding intersection is included in the iteration procedure for determining the parameters $\sqb{\lambda_i}$ with a reduced amount of overestimation. Finally, it should be pointed out that the presented simulation results already provide an important step toward the design of interval-based state of charge estimation strategies by means of measured terminal voltages that can be intersected with the corresponding simulated bounds in Figs.~\ref{fig_small_out} and~\ref{fig_large_out} in a predictor--corrector framework. They can be incorporated in battery management systems together with parameter identification, predictive control, and interval observer synthesis.

\section{Conclusions and Outlook on Future Work}\label{sec:concl}
In this paper, extensions of an interval-based exponential enclosure technique originally developed for integer-order sets of ordinary differential equations were presented. They are based on describing verified state enclosures with the help of Mittag-Leffler functions. In such a way, a simulation-based reachability analysis, in the sense of computing guaranteed outer interval enclosures, becomes possible for explicit, continuous-time sets of fractional-order differential equations. Commonly, the subdivision of the integration time horizon into shorter time slices allows for a reduction of overestimation. To exploit this property despite the infinite horizon memory of fractional-order differential equations, it is essential to rigorously quantify temporal truncation errors. 

Besides the automatic choice of optimal discretization step sizes, as a compromise between computational effort and tightness of the obtained state enclosures, future work will also address systems with oscillatory behavior, where complex-valued transformations already exist for the integer-order case~\cite{Rauh_SCAN14}. 

\bibliographystyle{eptcs}
\bibliography{mech_literatur}
\end{document}

%% file: defs_andreas.tex
\usepackage{ifpdf}
\usepackage{cite}
\usepackage{graphicx}
\usepackage[pdftex]{color}
\usepackage{amssymb}
\usepackage[centertags]{amsmath}
\usepackage{array}
\usepackage{amsfonts} 
\usepackage{bm}

\usepackage[caption=false]{subfig}

\usepackage{boondox-cal}

\usepackage{paralist}

\usepackage{xfrac}
\usepackage{tikz}
\usepackage{tikz,pgfplots}
\usetikzlibrary{arrows, intersections}
\makeatletter

\newcommand{\Rmnum}[1]{\expandafter\@slowromancap\romannumeral #1@}

\newcommand{\diag}[1]{\mathrm{diag}\left\{{#1}\right\}}

\newcommand{\rb}[1]{\left({#1}\right)}

\newcommand{\sqb}[1]{\left[{#1}\right]}

\newcommand{\sqbs}[1]{[{#1}]}

\newcommand{\intv}[2]{\left[{#1}\; ; \; {#2}\right]}

\newcommand{\ol}[1]{\overline{#1}}
\newcommand{\ul}[1]{\underline{#1}}

\newcommand{\dbr}[1]{\langle {#1} \rangle}

\usepackage{mathtools}

\newcommand\imCMsym[4][\mathord]{%
	\DeclareFontFamily{U} {#2}{}
	\DeclareFontShape{U}{#2}{m}{n}{
		<-6> #25
		<6-7> #26
		<7-8> #27
		<8-9> #28
		<9-10> #29
		<10-12> #210
		<12-> #212}{}
	\DeclareSymbolFont{CM#2} {U} {#2}{m}{n}
	\DeclareMathSymbol{#4}{#1}{CM#2}{#3}
}

\imCMsym{cmmi}{124}{\CMjmath}
\renewcommand{\jmath}{\CMjmath}

%% file: ex_nyquist/fig_ex_SNR20_ampl.pstex_t
\begin{picture}(0,0)%
\includegraphics{ex_nyquist/fig_ex_SNR20_ampl}%
\end{picture}%
\setlength{\unitlength}{4144sp}%
\begingroup\makeatletter\ifx\SetFigFont\undefined%
\gdef\SetFigFont#1#2#3#4#5{%
  \reset@font\fontsize{#1}{#2pt}%
  \fontfamily{#3}\fontseries{#4}\fontshape{#5}%
  \selectfont}%
\fi\endgroup%
\begin{picture}(6819,5424)(619,-5608)
\put(6166,-5476){\makebox(0,0)[rb]{\smash{{\SetFigFont{20}{24.0}{\rmdefault}{\mddefault}{\updefault}{\color[rgb]{0,0,0}$\omega$ in rad/s}%
}}}}
\put(946,-1321){\rotatebox{90.0}{\makebox(0,0)[rb]{\smash{{\SetFigFont{20}{24.0}{\rmdefault}{\mddefault}{\updefault}{\color[rgb]{0,0,0}magnitude in dB}%
}}}}}
\put(1756,-736){\makebox(0,0)[rb]{\smash{{\SetFigFont{20}{24.0}{\rmdefault}{\mddefault}{\updefault}{\color[rgb]{0,0,0}$0$}%
}}}}
\put(1756,-4246){\makebox(0,0)[rb]{\smash{{\SetFigFont{20}{24.0}{\rmdefault}{\mddefault}{\updefault}{\color[rgb]{0,0,0}$-30$}%
}}}}
\put(1756,-3076){\makebox(0,0)[rb]{\smash{{\SetFigFont{20}{24.0}{\rmdefault}{\mddefault}{\updefault}{\color[rgb]{0,0,0}$-20$}%
}}}}
\put(1756,-1861){\makebox(0,0)[rb]{\smash{{\SetFigFont{20}{24.0}{\rmdefault}{\mddefault}{\updefault}{\color[rgb]{0,0,0}$-10$}%
}}}}
\put(2656,-5056){\makebox(0,0)[b]{\smash{{\SetFigFont{20}{24.0}{\rmdefault}{\mddefault}{\updefault}{\color[rgb]{0,0,0}$\ul{\omega}$}%
}}}}
\put(5986,-5056){\makebox(0,0)[b]{\smash{{\SetFigFont{20}{24.0}{\rmdefault}{\mddefault}{\updefault}{\color[rgb]{0,0,0}$\ol{\omega}$}%
}}}}
\put(3511,-5056){\makebox(0,0)[b]{\smash{{\SetFigFont{20}{24.0}{\rmdefault}{\mddefault}{\updefault}{\color[rgb]{0,0,0}$10^{-1}$}%
}}}}
\put(1846,-5056){\makebox(0,0)[b]{\smash{{\SetFigFont{20}{24.0}{\rmdefault}{\mddefault}{\updefault}{\color[rgb]{0,0,0}$10^{-3}$}%
}}}}
\put(5176,-5056){\makebox(0,0)[b]{\smash{{\SetFigFont{20}{24.0}{\rmdefault}{\mddefault}{\updefault}{\color[rgb]{0,0,0}$10^1$}%
}}}}
\put(6841,-5056){\makebox(0,0)[b]{\smash{{\SetFigFont{20}{24.0}{\rmdefault}{\mddefault}{\updefault}{\color[rgb]{0,0,0}$10^3$}%
}}}}
\end{picture}%

%% file: ex_nyquist/fig_ex_SNR20_phase.pstex_t
\begin{picture}(0,0)%
\includegraphics{ex_nyquist/fig_ex_SNR20_phase}%
\end{picture}%
\setlength{\unitlength}{4144sp}%
\begingroup\makeatletter\ifx\SetFigFont\undefined%
\gdef\SetFigFont#1#2#3#4#5{%
  \reset@font\fontsize{#1}{#2pt}%
  \fontfamily{#3}\fontseries{#4}\fontshape{#5}%
  \selectfont}%
\fi\endgroup%
\begin{picture}(6819,5424)(619,-5608)
\put(6166,-5476){\makebox(0,0)[rb]{\smash{{\SetFigFont{20}{24.0}{\rmdefault}{\mddefault}{\updefault}{\color[rgb]{0,0,0}$\omega$ in rad/s}%
}}}}
\put(946,-1321){\rotatebox{90.0}{\makebox(0,0)[rb]{\smash{{\SetFigFont{20}{24.0}{\rmdefault}{\mddefault}{\updefault}{\color[rgb]{0,0,0}phase in rad}%
}}}}}
\put(1756,-736){\makebox(0,0)[rb]{\smash{{\SetFigFont{20}{24.0}{\rmdefault}{\mddefault}{\updefault}{\color[rgb]{0,0,0}$0$}%
}}}}
\put(1756,-1726){\makebox(0,0)[rb]{\smash{{\SetFigFont{20}{24.0}{\rmdefault}{\mddefault}{\updefault}{\color[rgb]{0,0,0}$-0.2$}%
}}}}
\put(1756,-2761){\makebox(0,0)[rb]{\smash{{\SetFigFont{20}{24.0}{\rmdefault}{\mddefault}{\updefault}{\color[rgb]{0,0,0}$-0.4$}%
}}}}
\put(1756,-3796){\makebox(0,0)[rb]{\smash{{\SetFigFont{20}{24.0}{\rmdefault}{\mddefault}{\updefault}{\color[rgb]{0,0,0}$-0.6$}%
}}}}
\put(1711,-4786){\makebox(0,0)[rb]{\smash{{\SetFigFont{20}{24.0}{\rmdefault}{\mddefault}{\updefault}{\color[rgb]{0,0,0}$-0.8$}%
}}}}
\put(2656,-5056){\makebox(0,0)[b]{\smash{{\SetFigFont{20}{24.0}{\rmdefault}{\mddefault}{\updefault}{\color[rgb]{0,0,0}$\ul{\omega}$}%
}}}}
\put(5986,-5056){\makebox(0,0)[b]{\smash{{\SetFigFont{20}{24.0}{\rmdefault}{\mddefault}{\updefault}{\color[rgb]{0,0,0}$\ol{\omega}$}%
}}}}
\put(1846,-5056){\makebox(0,0)[b]{\smash{{\SetFigFont{20}{24.0}{\rmdefault}{\mddefault}{\updefault}{\color[rgb]{0,0,0}$10^{-3}$}%
}}}}
\put(3511,-5056){\makebox(0,0)[b]{\smash{{\SetFigFont{20}{24.0}{\rmdefault}{\mddefault}{\updefault}{\color[rgb]{0,0,0}$10^{-1}$}%
}}}}
\put(5176,-5056){\makebox(0,0)[b]{\smash{{\SetFigFont{20}{24.0}{\rmdefault}{\mddefault}{\updefault}{\color[rgb]{0,0,0}$10^1$}%
}}}}
\put(6841,-5056){\makebox(0,0)[b]{\smash{{\SetFigFont{20}{24.0}{\rmdefault}{\mddefault}{\updefault}{\color[rgb]{0,0,0}$10^3$}%
}}}}
\end{picture}%

%% file: ex_nyquist/fig_ex_SNR20_nyquist.pstex_t
\begin{picture}(0,0)%
\includegraphics{ex_nyquist/fig_ex_SNR20_nyquist}%
\end{picture}%
\setlength{\unitlength}{4144sp}%
\begingroup\makeatletter\ifx\SetFigFont\undefined%
\gdef\SetFigFont#1#2#3#4#5{%
  \reset@font\fontsize{#1}{#2pt}%
  \fontfamily{#3}\fontseries{#4}\fontshape{#5}%
  \selectfont}%
\fi\endgroup%
\begin{picture}(6819,5424)(619,-5608)
\put(6166,-5476){\makebox(0,0)[rb]{\smash{{\SetFigFont{20}{24.0}{\rmdefault}{\mddefault}{\updefault}{\color[rgb]{0,0,0}$\Re\{F(\jmath \omega)\}$}%
}}}}
\put(946,-1321){\rotatebox{90.0}{\makebox(0,0)[rb]{\smash{{\SetFigFont{20}{24.0}{\rmdefault}{\mddefault}{\updefault}{\color[rgb]{0,0,0}$\Im\{F(\jmath \omega)\}$}%
}}}}}
\put(1756,-4291){\makebox(0,0)[rb]{\smash{{\SetFigFont{20}{24.0}{\rmdefault}{\mddefault}{\updefault}{\color[rgb]{0,0,0}$-0.4$}%
}}}}
\put(1756,-3166){\makebox(0,0)[rb]{\smash{{\SetFigFont{20}{24.0}{\rmdefault}{\mddefault}{\updefault}{\color[rgb]{0,0,0}$-0.2$}%
}}}}
\put(1756,-2086){\makebox(0,0)[rb]{\smash{{\SetFigFont{20}{24.0}{\rmdefault}{\mddefault}{\updefault}{\color[rgb]{0,0,0}$0$}%
}}}}
\put(1756,-1006){\makebox(0,0)[rb]{\smash{{\SetFigFont{20}{24.0}{\rmdefault}{\mddefault}{\updefault}{\color[rgb]{0,0,0}$0.2$}%
}}}}
\put(2746,-5056){\makebox(0,0)[b]{\smash{{\SetFigFont{20}{24.0}{\rmdefault}{\mddefault}{\updefault}{\color[rgb]{0,0,0}$0.2$}%
}}}}
\put(3826,-5056){\makebox(0,0)[b]{\smash{{\SetFigFont{20}{24.0}{\rmdefault}{\mddefault}{\updefault}{\color[rgb]{0,0,0}$0.4$}%
}}}}
\put(4861,-5056){\makebox(0,0)[b]{\smash{{\SetFigFont{20}{24.0}{\rmdefault}{\mddefault}{\updefault}{\color[rgb]{0,0,0}$0.6$}%
}}}}
\put(5896,-5056){\makebox(0,0)[b]{\smash{{\SetFigFont{20}{24.0}{\rmdefault}{\mddefault}{\updefault}{\color[rgb]{0,0,0}$0.8$}%
}}}}
\end{picture}%

%% file: fig_SNR/fig_SNR_klein_1.pstex_t
\begin{picture}(0,0)%
\includegraphics{fig_SNR/fig_SNR_klein_1}%
\end{picture}%
\setlength{\unitlength}{4144sp}%
\begingroup\makeatletter\ifx\SetFigFont\undefined%
\gdef\SetFigFont#1#2#3#4#5{%
  \reset@font\fontsize{#1}{#2pt}%
  \fontfamily{#3}\fontseries{#4}\fontshape{#5}%
  \selectfont}%
\fi\endgroup%
\begin{picture}(6819,5406)(619,-5590)
\put(946,-1321){\rotatebox{90.0}{\makebox(0,0)[rb]{\smash{{\SetFigFont{20}{24.0}{\rmdefault}{\mddefault}{\updefault}{\color[rgb]{0,0,0}$\sqb{z}(t)$}%
}}}}}
\put(6166,-5476){\makebox(0,0)[rb]{\smash{{\SetFigFont{20}{24.0}{\rmdefault}{\mddefault}{\updefault}{\color[rgb]{0,0,0}$t$}%
}}}}
\put(1846,-5056){\makebox(0,0)[b]{\smash{{\SetFigFont{20}{24.0}{\rmdefault}{\mddefault}{\updefault}{\color[rgb]{0,0,0}$0$}%
}}}}
\put(1756,-4786){\makebox(0,0)[rb]{\smash{{\SetFigFont{20}{24.0}{\rmdefault}{\mddefault}{\updefault}{\color[rgb]{0,0,0}$0$}%
}}}}
\put(1756,-3166){\makebox(0,0)[rb]{\smash{{\SetFigFont{20}{24.0}{\rmdefault}{\mddefault}{\updefault}{\color[rgb]{0,0,0}$0.4$}%
}}}}
\put(1756,-4021){\makebox(0,0)[rb]{\smash{{\SetFigFont{20}{24.0}{\rmdefault}{\mddefault}{\updefault}{\color[rgb]{0,0,0}$0.2$}%
}}}}
\put(1756,-2311){\makebox(0,0)[rb]{\smash{{\SetFigFont{20}{24.0}{\rmdefault}{\mddefault}{\updefault}{\color[rgb]{0,0,0}$0.6$}%
}}}}
\put(1756,-736){\makebox(0,0)[rb]{\smash{{\SetFigFont{20}{24.0}{\rmdefault}{\mddefault}{\updefault}{\color[rgb]{0,0,0}$1.0$}%
}}}}
\put(1756,-1501){\makebox(0,0)[rb]{\smash{{\SetFigFont{20}{24.0}{\rmdefault}{\mddefault}{\updefault}{\color[rgb]{0,0,0}$0.8$}%
}}}}
\put(2836,-5056){\makebox(0,0)[b]{\smash{{\SetFigFont{20}{24.0}{\rmdefault}{\mddefault}{\updefault}{\color[rgb]{0,0,0}$0.2$}%
}}}}
\put(3826,-5056){\makebox(0,0)[b]{\smash{{\SetFigFont{20}{24.0}{\rmdefault}{\mddefault}{\updefault}{\color[rgb]{0,0,0}$0.4$}%
}}}}
\put(4816,-5056){\makebox(0,0)[b]{\smash{{\SetFigFont{20}{24.0}{\rmdefault}{\mddefault}{\updefault}{\color[rgb]{0,0,0}$0.6$}%
}}}}
\put(5851,-5056){\makebox(0,0)[b]{\smash{{\SetFigFont{20}{24.0}{\rmdefault}{\mddefault}{\updefault}{\color[rgb]{0,0,0}$0.8$}%
}}}}
\put(6841,-5056){\makebox(0,0)[b]{\smash{{\SetFigFont{20}{24.0}{\rmdefault}{\mddefault}{\updefault}{\color[rgb]{0,0,0}$1.0$}%
}}}}
\end{picture}%

%% file: fig_SNR/fig_SNR_gross_1.pstex_t
\begin{picture}(0,0)%
\includegraphics{fig_SNR/fig_SNR_gross_1}%
\end{picture}%
\setlength{\unitlength}{4144sp}%
\begingroup\makeatletter\ifx\SetFigFont\undefined%
\gdef\SetFigFont#1#2#3#4#5{%
  \reset@font\fontsize{#1}{#2pt}%
  \fontfamily{#3}\fontseries{#4}\fontshape{#5}%
  \selectfont}%
\fi\endgroup%
\begin{picture}(6819,5406)(619,-5590)
\put(946,-1321){\rotatebox{90.0}{\makebox(0,0)[rb]{\smash{{\SetFigFont{20}{24.0}{\rmdefault}{\mddefault}{\updefault}{\color[rgb]{0,0,0}$\sqb{z}(t)$}%
}}}}}
\put(6166,-5476){\makebox(0,0)[rb]{\smash{{\SetFigFont{20}{24.0}{\rmdefault}{\mddefault}{\updefault}{\color[rgb]{0,0,0}$t$}%
}}}}
\put(1846,-5056){\makebox(0,0)[b]{\smash{{\SetFigFont{20}{24.0}{\rmdefault}{\mddefault}{\updefault}{\color[rgb]{0,0,0}$0$}%
}}}}
\put(1756,-4786){\makebox(0,0)[rb]{\smash{{\SetFigFont{20}{24.0}{\rmdefault}{\mddefault}{\updefault}{\color[rgb]{0,0,0}$0$}%
}}}}
\put(1756,-3166){\makebox(0,0)[rb]{\smash{{\SetFigFont{20}{24.0}{\rmdefault}{\mddefault}{\updefault}{\color[rgb]{0,0,0}$0.4$}%
}}}}
\put(1756,-4021){\makebox(0,0)[rb]{\smash{{\SetFigFont{20}{24.0}{\rmdefault}{\mddefault}{\updefault}{\color[rgb]{0,0,0}$0.2$}%
}}}}
\put(1756,-2311){\makebox(0,0)[rb]{\smash{{\SetFigFont{20}{24.0}{\rmdefault}{\mddefault}{\updefault}{\color[rgb]{0,0,0}$0.6$}%
}}}}
\put(1756,-736){\makebox(0,0)[rb]{\smash{{\SetFigFont{20}{24.0}{\rmdefault}{\mddefault}{\updefault}{\color[rgb]{0,0,0}$1.0$}%
}}}}
\put(1756,-1501){\makebox(0,0)[rb]{\smash{{\SetFigFont{20}{24.0}{\rmdefault}{\mddefault}{\updefault}{\color[rgb]{0,0,0}$0.8$}%
}}}}
\put(2836,-5056){\makebox(0,0)[b]{\smash{{\SetFigFont{20}{24.0}{\rmdefault}{\mddefault}{\updefault}{\color[rgb]{0,0,0}$0.2$}%
}}}}
\put(3826,-5056){\makebox(0,0)[b]{\smash{{\SetFigFont{20}{24.0}{\rmdefault}{\mddefault}{\updefault}{\color[rgb]{0,0,0}$0.4$}%
}}}}
\put(4816,-5056){\makebox(0,0)[b]{\smash{{\SetFigFont{20}{24.0}{\rmdefault}{\mddefault}{\updefault}{\color[rgb]{0,0,0}$0.6$}%
}}}}
\put(5851,-5056){\makebox(0,0)[b]{\smash{{\SetFigFont{20}{24.0}{\rmdefault}{\mddefault}{\updefault}{\color[rgb]{0,0,0}$0.8$}%
}}}}
\put(6841,-5056){\makebox(0,0)[b]{\smash{{\SetFigFont{20}{24.0}{\rmdefault}{\mddefault}{\updefault}{\color[rgb]{0,0,0}$1.0$}%
}}}}
\end{picture}%

%% file: fig_SNR/fig_SNR_klein_2.pstex_t
\begin{picture}(0,0)%
\includegraphics{fig_SNR/fig_SNR_klein_2}%
\end{picture}%
\setlength{\unitlength}{4144sp}%
\begingroup\makeatletter\ifx\SetFigFont\undefined%
\gdef\SetFigFont#1#2#3#4#5{%
  \reset@font\fontsize{#1}{#2pt}%
  \fontfamily{#3}\fontseries{#4}\fontshape{#5}%
  \selectfont}%
\fi\endgroup%
\begin{picture}(6819,5406)(619,-5590)
\put(946,-1321){\rotatebox{90.0}{\makebox(0,0)[rb]{\smash{{\SetFigFont{20}{24.0}{\rmdefault}{\mddefault}{\updefault}{\color[rgb]{0,0,0}$\sqb{z}(t)$}%
}}}}}
\put(6166,-5476){\makebox(0,0)[rb]{\smash{{\SetFigFont{20}{24.0}{\rmdefault}{\mddefault}{\updefault}{\color[rgb]{0,0,0}$t$}%
}}}}
\put(1846,-5056){\makebox(0,0)[b]{\smash{{\SetFigFont{20}{24.0}{\rmdefault}{\mddefault}{\updefault}{\color[rgb]{0,0,0}$0$}%
}}}}
\put(1756,-4786){\makebox(0,0)[rb]{\smash{{\SetFigFont{20}{24.0}{\rmdefault}{\mddefault}{\updefault}{\color[rgb]{0,0,0}$0$}%
}}}}
\put(1756,-3166){\makebox(0,0)[rb]{\smash{{\SetFigFont{20}{24.0}{\rmdefault}{\mddefault}{\updefault}{\color[rgb]{0,0,0}$0.4$}%
}}}}
\put(1756,-4021){\makebox(0,0)[rb]{\smash{{\SetFigFont{20}{24.0}{\rmdefault}{\mddefault}{\updefault}{\color[rgb]{0,0,0}$0.2$}%
}}}}
\put(1756,-2311){\makebox(0,0)[rb]{\smash{{\SetFigFont{20}{24.0}{\rmdefault}{\mddefault}{\updefault}{\color[rgb]{0,0,0}$0.6$}%
}}}}
\put(1756,-736){\makebox(0,0)[rb]{\smash{{\SetFigFont{20}{24.0}{\rmdefault}{\mddefault}{\updefault}{\color[rgb]{0,0,0}$1.0$}%
}}}}
\put(1756,-1501){\makebox(0,0)[rb]{\smash{{\SetFigFont{20}{24.0}{\rmdefault}{\mddefault}{\updefault}{\color[rgb]{0,0,0}$0.8$}%
}}}}
\put(2836,-5056){\makebox(0,0)[b]{\smash{{\SetFigFont{20}{24.0}{\rmdefault}{\mddefault}{\updefault}{\color[rgb]{0,0,0}$0.2$}%
}}}}
\put(3826,-5056){\makebox(0,0)[b]{\smash{{\SetFigFont{20}{24.0}{\rmdefault}{\mddefault}{\updefault}{\color[rgb]{0,0,0}$0.4$}%
}}}}
\put(4816,-5056){\makebox(0,0)[b]{\smash{{\SetFigFont{20}{24.0}{\rmdefault}{\mddefault}{\updefault}{\color[rgb]{0,0,0}$0.6$}%
}}}}
\put(5851,-5056){\makebox(0,0)[b]{\smash{{\SetFigFont{20}{24.0}{\rmdefault}{\mddefault}{\updefault}{\color[rgb]{0,0,0}$0.8$}%
}}}}
\put(6841,-5056){\makebox(0,0)[b]{\smash{{\SetFigFont{20}{24.0}{\rmdefault}{\mddefault}{\updefault}{\color[rgb]{0,0,0}$1.0$}%
}}}}
\end{picture}%

%% file: fig_SNR/fig_SNR_gross_2.pstex_t
\begin{picture}(0,0)%
\includegraphics{fig_SNR/fig_SNR_gross_2}%
\end{picture}%
\setlength{\unitlength}{4144sp}%
\begingroup\makeatletter\ifx\SetFigFont\undefined%
\gdef\SetFigFont#1#2#3#4#5{%
  \reset@font\fontsize{#1}{#2pt}%
  \fontfamily{#3}\fontseries{#4}\fontshape{#5}%
  \selectfont}%
\fi\endgroup%
\begin{picture}(6819,5406)(619,-5590)
\put(946,-1321){\rotatebox{90.0}{\makebox(0,0)[rb]{\smash{{\SetFigFont{20}{24.0}{\rmdefault}{\mddefault}{\updefault}{\color[rgb]{0,0,0}$\sqb{z}(t)$}%
}}}}}
\put(6166,-5476){\makebox(0,0)[rb]{\smash{{\SetFigFont{20}{24.0}{\rmdefault}{\mddefault}{\updefault}{\color[rgb]{0,0,0}$t$}%
}}}}
\put(1846,-5056){\makebox(0,0)[b]{\smash{{\SetFigFont{20}{24.0}{\rmdefault}{\mddefault}{\updefault}{\color[rgb]{0,0,0}$0$}%
}}}}
\put(1756,-4786){\makebox(0,0)[rb]{\smash{{\SetFigFont{20}{24.0}{\rmdefault}{\mddefault}{\updefault}{\color[rgb]{0,0,0}$0$}%
}}}}
\put(1756,-3166){\makebox(0,0)[rb]{\smash{{\SetFigFont{20}{24.0}{\rmdefault}{\mddefault}{\updefault}{\color[rgb]{0,0,0}$0.4$}%
}}}}
\put(1756,-4021){\makebox(0,0)[rb]{\smash{{\SetFigFont{20}{24.0}{\rmdefault}{\mddefault}{\updefault}{\color[rgb]{0,0,0}$0.2$}%
}}}}
\put(1756,-2311){\makebox(0,0)[rb]{\smash{{\SetFigFont{20}{24.0}{\rmdefault}{\mddefault}{\updefault}{\color[rgb]{0,0,0}$0.6$}%
}}}}
\put(1756,-736){\makebox(0,0)[rb]{\smash{{\SetFigFont{20}{24.0}{\rmdefault}{\mddefault}{\updefault}{\color[rgb]{0,0,0}$1.0$}%
}}}}
\put(1756,-1501){\makebox(0,0)[rb]{\smash{{\SetFigFont{20}{24.0}{\rmdefault}{\mddefault}{\updefault}{\color[rgb]{0,0,0}$0.8$}%
}}}}
\put(2836,-5056){\makebox(0,0)[b]{\smash{{\SetFigFont{20}{24.0}{\rmdefault}{\mddefault}{\updefault}{\color[rgb]{0,0,0}$0.2$}%
}}}}
\put(3826,-5056){\makebox(0,0)[b]{\smash{{\SetFigFont{20}{24.0}{\rmdefault}{\mddefault}{\updefault}{\color[rgb]{0,0,0}$0.4$}%
}}}}
\put(4816,-5056){\makebox(0,0)[b]{\smash{{\SetFigFont{20}{24.0}{\rmdefault}{\mddefault}{\updefault}{\color[rgb]{0,0,0}$0.6$}%
}}}}
\put(5851,-5056){\makebox(0,0)[b]{\smash{{\SetFigFont{20}{24.0}{\rmdefault}{\mddefault}{\updefault}{\color[rgb]{0,0,0}$0.8$}%
}}}}
\put(6841,-5056){\makebox(0,0)[b]{\smash{{\SetFigFont{20}{24.0}{\rmdefault}{\mddefault}{\updefault}{\color[rgb]{0,0,0}$1.0$}%
}}}}
\end{picture}%

%% file: fig_SNR/fig_SNR_klein_3.pstex_t
\begin{picture}(0,0)%
\includegraphics{fig_SNR/fig_SNR_klein_3}%
\end{picture}%
\setlength{\unitlength}{4144sp}%
\begingroup\makeatletter\ifx\SetFigFont\undefined%
\gdef\SetFigFont#1#2#3#4#5{%
  \reset@font\fontsize{#1}{#2pt}%
  \fontfamily{#3}\fontseries{#4}\fontshape{#5}%
  \selectfont}%
\fi\endgroup%
\begin{picture}(6819,5406)(619,-5590)
\put(946,-1321){\rotatebox{90.0}{\makebox(0,0)[rb]{\smash{{\SetFigFont{20}{24.0}{\rmdefault}{\mddefault}{\updefault}{\color[rgb]{0,0,0}$\sqb{z}(t)$}%
}}}}}
\put(6166,-5476){\makebox(0,0)[rb]{\smash{{\SetFigFont{20}{24.0}{\rmdefault}{\mddefault}{\updefault}{\color[rgb]{0,0,0}$t$}%
}}}}
\put(1846,-5056){\makebox(0,0)[b]{\smash{{\SetFigFont{20}{24.0}{\rmdefault}{\mddefault}{\updefault}{\color[rgb]{0,0,0}$0$}%
}}}}
\put(1756,-4786){\makebox(0,0)[rb]{\smash{{\SetFigFont{20}{24.0}{\rmdefault}{\mddefault}{\updefault}{\color[rgb]{0,0,0}$0$}%
}}}}
\put(1756,-3166){\makebox(0,0)[rb]{\smash{{\SetFigFont{20}{24.0}{\rmdefault}{\mddefault}{\updefault}{\color[rgb]{0,0,0}$0.4$}%
}}}}
\put(1756,-4021){\makebox(0,0)[rb]{\smash{{\SetFigFont{20}{24.0}{\rmdefault}{\mddefault}{\updefault}{\color[rgb]{0,0,0}$0.2$}%
}}}}
\put(1756,-2311){\makebox(0,0)[rb]{\smash{{\SetFigFont{20}{24.0}{\rmdefault}{\mddefault}{\updefault}{\color[rgb]{0,0,0}$0.6$}%
}}}}
\put(1756,-736){\makebox(0,0)[rb]{\smash{{\SetFigFont{20}{24.0}{\rmdefault}{\mddefault}{\updefault}{\color[rgb]{0,0,0}$1.0$}%
}}}}
\put(1756,-1501){\makebox(0,0)[rb]{\smash{{\SetFigFont{20}{24.0}{\rmdefault}{\mddefault}{\updefault}{\color[rgb]{0,0,0}$0.8$}%
}}}}
\put(2836,-5056){\makebox(0,0)[b]{\smash{{\SetFigFont{20}{24.0}{\rmdefault}{\mddefault}{\updefault}{\color[rgb]{0,0,0}$0.2$}%
}}}}
\put(3826,-5056){\makebox(0,0)[b]{\smash{{\SetFigFont{20}{24.0}{\rmdefault}{\mddefault}{\updefault}{\color[rgb]{0,0,0}$0.4$}%
}}}}
\put(4816,-5056){\makebox(0,0)[b]{\smash{{\SetFigFont{20}{24.0}{\rmdefault}{\mddefault}{\updefault}{\color[rgb]{0,0,0}$0.6$}%
}}}}
\put(5851,-5056){\makebox(0,0)[b]{\smash{{\SetFigFont{20}{24.0}{\rmdefault}{\mddefault}{\updefault}{\color[rgb]{0,0,0}$0.8$}%
}}}}
\put(6841,-5056){\makebox(0,0)[b]{\smash{{\SetFigFont{20}{24.0}{\rmdefault}{\mddefault}{\updefault}{\color[rgb]{0,0,0}$1.0$}%
}}}}
\end{picture}%

%% file: fig_SNR/fig_SNR_gross_3.pstex_t
\begin{picture}(0,0)%
\includegraphics{fig_SNR/fig_SNR_gross_3}%
\end{picture}%
\setlength{\unitlength}{4144sp}%
\begingroup\makeatletter\ifx\SetFigFont\undefined%
\gdef\SetFigFont#1#2#3#4#5{%
  \reset@font\fontsize{#1}{#2pt}%
  \fontfamily{#3}\fontseries{#4}\fontshape{#5}%
  \selectfont}%
\fi\endgroup%
\begin{picture}(6819,5406)(619,-5590)
\put(946,-1321){\rotatebox{90.0}{\makebox(0,0)[rb]{\smash{{\SetFigFont{20}{24.0}{\rmdefault}{\mddefault}{\updefault}{\color[rgb]{0,0,0}$\sqb{z}(t)$}%
}}}}}
\put(6166,-5476){\makebox(0,0)[rb]{\smash{{\SetFigFont{20}{24.0}{\rmdefault}{\mddefault}{\updefault}{\color[rgb]{0,0,0}$t$}%
}}}}
\put(1846,-5056){\makebox(0,0)[b]{\smash{{\SetFigFont{20}{24.0}{\rmdefault}{\mddefault}{\updefault}{\color[rgb]{0,0,0}$0$}%
}}}}
\put(1756,-4786){\makebox(0,0)[rb]{\smash{{\SetFigFont{20}{24.0}{\rmdefault}{\mddefault}{\updefault}{\color[rgb]{0,0,0}$0$}%
}}}}
\put(1756,-3166){\makebox(0,0)[rb]{\smash{{\SetFigFont{20}{24.0}{\rmdefault}{\mddefault}{\updefault}{\color[rgb]{0,0,0}$0.4$}%
}}}}
\put(1756,-4021){\makebox(0,0)[rb]{\smash{{\SetFigFont{20}{24.0}{\rmdefault}{\mddefault}{\updefault}{\color[rgb]{0,0,0}$0.2$}%
}}}}
\put(1756,-2311){\makebox(0,0)[rb]{\smash{{\SetFigFont{20}{24.0}{\rmdefault}{\mddefault}{\updefault}{\color[rgb]{0,0,0}$0.6$}%
}}}}
\put(1756,-736){\makebox(0,0)[rb]{\smash{{\SetFigFont{20}{24.0}{\rmdefault}{\mddefault}{\updefault}{\color[rgb]{0,0,0}$1.0$}%
}}}}
\put(1756,-1501){\makebox(0,0)[rb]{\smash{{\SetFigFont{20}{24.0}{\rmdefault}{\mddefault}{\updefault}{\color[rgb]{0,0,0}$0.8$}%
}}}}
\put(2836,-5056){\makebox(0,0)[b]{\smash{{\SetFigFont{20}{24.0}{\rmdefault}{\mddefault}{\updefault}{\color[rgb]{0,0,0}$0.2$}%
}}}}
\put(3826,-5056){\makebox(0,0)[b]{\smash{{\SetFigFont{20}{24.0}{\rmdefault}{\mddefault}{\updefault}{\color[rgb]{0,0,0}$0.4$}%
}}}}
\put(4816,-5056){\makebox(0,0)[b]{\smash{{\SetFigFont{20}{24.0}{\rmdefault}{\mddefault}{\updefault}{\color[rgb]{0,0,0}$0.6$}%
}}}}
\put(5851,-5056){\makebox(0,0)[b]{\smash{{\SetFigFont{20}{24.0}{\rmdefault}{\mddefault}{\updefault}{\color[rgb]{0,0,0}$0.8$}%
}}}}
\put(6841,-5056){\makebox(0,0)[b]{\smash{{\SetFigFont{20}{24.0}{\rmdefault}{\mddefault}{\updefault}{\color[rgb]{0,0,0}$1.0$}%
}}}}
\end{picture}%

%% file: fig_SNR/fig_SOC_small.pstex_t
\begin{picture}(0,0)%
\includegraphics{fig_SNR/fig_SOC_small}%
\end{picture}%
\setlength{\unitlength}{4144sp}%
\begingroup\makeatletter\ifx\SetFigFont\undefined%
\gdef\SetFigFont#1#2#3#4#5{%
  \reset@font\fontsize{#1}{#2pt}%
  \fontfamily{#3}\fontseries{#4}\fontshape{#5}%
  \selectfont}%
\fi\endgroup%
\begin{picture}(6819,5406)(619,-5590)
\put(6166,-5476){\makebox(0,0)[rb]{\smash{{\SetFigFont{20}{24.0}{\rmdefault}{\mddefault}{\updefault}{\color[rgb]{0,0,0}$t$ in s}%
}}}}
\put(946,-1321){\rotatebox{90.0}{\makebox(0,0)[rb]{\smash{{\SetFigFont{20}{24.0}{\rmdefault}{\mddefault}{\updefault}{\color[rgb]{0,0,0}$\sqb{\sigma}(t)$}%
}}}}}
\put(1846,-5056){\makebox(0,0)[b]{\smash{{\SetFigFont{20}{24.0}{\rmdefault}{\mddefault}{\updefault}{\color[rgb]{0,0,0}$0$}%
}}}}
\put(2836,-5056){\makebox(0,0)[b]{\smash{{\SetFigFont{20}{24.0}{\rmdefault}{\mddefault}{\updefault}{\color[rgb]{0,0,0}$2$}%
}}}}
\put(3826,-5056){\makebox(0,0)[b]{\smash{{\SetFigFont{20}{24.0}{\rmdefault}{\mddefault}{\updefault}{\color[rgb]{0,0,0}$4$}%
}}}}
\put(4816,-5056){\makebox(0,0)[b]{\smash{{\SetFigFont{20}{24.0}{\rmdefault}{\mddefault}{\updefault}{\color[rgb]{0,0,0}$6$}%
}}}}
\put(5851,-5056){\makebox(0,0)[b]{\smash{{\SetFigFont{20}{24.0}{\rmdefault}{\mddefault}{\updefault}{\color[rgb]{0,0,0}$8$}%
}}}}
\put(6841,-5056){\makebox(0,0)[b]{\smash{{\SetFigFont{20}{24.0}{\rmdefault}{\mddefault}{\updefault}{\color[rgb]{0,0,0}$10$}%
}}}}
\put(1756,-4786){\makebox(0,0)[rb]{\smash{{\SetFigFont{20}{24.0}{\rmdefault}{\mddefault}{\updefault}{\color[rgb]{0,0,0}$0.85$}%
}}}}
\put(1756,-2761){\makebox(0,0)[rb]{\smash{{\SetFigFont{20}{24.0}{\rmdefault}{\mddefault}{\updefault}{\color[rgb]{0,0,0}$1.00$}%
}}}}
\put(1756,-2086){\makebox(0,0)[rb]{\smash{{\SetFigFont{20}{24.0}{\rmdefault}{\mddefault}{\updefault}{\color[rgb]{0,0,0}$1.05$}%
}}}}
\put(1756,-1411){\makebox(0,0)[rb]{\smash{{\SetFigFont{20}{24.0}{\rmdefault}{\mddefault}{\updefault}{\color[rgb]{0,0,0}$1.10$}%
}}}}
\put(1756,-4111){\makebox(0,0)[rb]{\smash{{\SetFigFont{20}{24.0}{\rmdefault}{\mddefault}{\updefault}{\color[rgb]{0,0,0}$0.90$}%
}}}}
\put(1756,-3436){\makebox(0,0)[rb]{\smash{{\SetFigFont{20}{24.0}{\rmdefault}{\mddefault}{\updefault}{\color[rgb]{0,0,0}$0.95$}%
}}}}
\put(1756,-736){\makebox(0,0)[rb]{\smash{{\SetFigFont{20}{24.0}{\rmdefault}{\mddefault}{\updefault}{\color[rgb]{0,0,0}$1.15$}%
}}}}
\end{picture}%

%% file: fig_SNR/fig_SOC_large.pstex_t
\begin{picture}(0,0)%
\includegraphics{fig_SNR/fig_SOC_large}%
\end{picture}%
\setlength{\unitlength}{4144sp}%
\begingroup\makeatletter\ifx\SetFigFont\undefined%
\gdef\SetFigFont#1#2#3#4#5{%
  \reset@font\fontsize{#1}{#2pt}%
  \fontfamily{#3}\fontseries{#4}\fontshape{#5}%
  \selectfont}%
\fi\endgroup%
\begin{picture}(6819,5406)(619,-5590)
\put(6166,-5476){\makebox(0,0)[rb]{\smash{{\SetFigFont{20}{24.0}{\rmdefault}{\mddefault}{\updefault}{\color[rgb]{0,0,0}$t$ in s}%
}}}}
\put(946,-1321){\rotatebox{90.0}{\makebox(0,0)[rb]{\smash{{\SetFigFont{20}{24.0}{\rmdefault}{\mddefault}{\updefault}{\color[rgb]{0,0,0}$\sqb{\sigma}(t)$}%
}}}}}
\put(1846,-5056){\makebox(0,0)[b]{\smash{{\SetFigFont{20}{24.0}{\rmdefault}{\mddefault}{\updefault}{\color[rgb]{0,0,0}$0$}%
}}}}
\put(2836,-5056){\makebox(0,0)[b]{\smash{{\SetFigFont{20}{24.0}{\rmdefault}{\mddefault}{\updefault}{\color[rgb]{0,0,0}$2$}%
}}}}
\put(3826,-5056){\makebox(0,0)[b]{\smash{{\SetFigFont{20}{24.0}{\rmdefault}{\mddefault}{\updefault}{\color[rgb]{0,0,0}$4$}%
}}}}
\put(4816,-5056){\makebox(0,0)[b]{\smash{{\SetFigFont{20}{24.0}{\rmdefault}{\mddefault}{\updefault}{\color[rgb]{0,0,0}$6$}%
}}}}
\put(5851,-5056){\makebox(0,0)[b]{\smash{{\SetFigFont{20}{24.0}{\rmdefault}{\mddefault}{\updefault}{\color[rgb]{0,0,0}$8$}%
}}}}
\put(6841,-5056){\makebox(0,0)[b]{\smash{{\SetFigFont{20}{24.0}{\rmdefault}{\mddefault}{\updefault}{\color[rgb]{0,0,0}$10$}%
}}}}
\put(1756,-4786){\makebox(0,0)[rb]{\smash{{\SetFigFont{20}{24.0}{\rmdefault}{\mddefault}{\updefault}{\color[rgb]{0,0,0}$0.85$}%
}}}}
\put(1756,-2761){\makebox(0,0)[rb]{\smash{{\SetFigFont{20}{24.0}{\rmdefault}{\mddefault}{\updefault}{\color[rgb]{0,0,0}$1.00$}%
}}}}
\put(1756,-2086){\makebox(0,0)[rb]{\smash{{\SetFigFont{20}{24.0}{\rmdefault}{\mddefault}{\updefault}{\color[rgb]{0,0,0}$1.05$}%
}}}}
\put(1756,-1411){\makebox(0,0)[rb]{\smash{{\SetFigFont{20}{24.0}{\rmdefault}{\mddefault}{\updefault}{\color[rgb]{0,0,0}$1.10$}%
}}}}
\put(1756,-4111){\makebox(0,0)[rb]{\smash{{\SetFigFont{20}{24.0}{\rmdefault}{\mddefault}{\updefault}{\color[rgb]{0,0,0}$0.90$}%
}}}}
\put(1756,-3436){\makebox(0,0)[rb]{\smash{{\SetFigFont{20}{24.0}{\rmdefault}{\mddefault}{\updefault}{\color[rgb]{0,0,0}$0.95$}%
}}}}
\put(1756,-736){\makebox(0,0)[rb]{\smash{{\SetFigFont{20}{24.0}{\rmdefault}{\mddefault}{\updefault}{\color[rgb]{0,0,0}$1.15$}%
}}}}
\end{picture}%

%% file: fig_SNR/fig_VCPE_small.pstex_t
\begin{picture}(0,0)%
\includegraphics{fig_SNR/fig_VCPE_small}%
\end{picture}%
\setlength{\unitlength}{4144sp}%
\begingroup\makeatletter\ifx\SetFigFont\undefined%
\gdef\SetFigFont#1#2#3#4#5{%
  \reset@font\fontsize{#1}{#2pt}%
  \fontfamily{#3}\fontseries{#4}\fontshape{#5}%
  \selectfont}%
\fi\endgroup%
\begin{picture}(6819,5406)(619,-5590)
\put(6166,-5476){\makebox(0,0)[rb]{\smash{{\SetFigFont{20}{24.0}{\rmdefault}{\mddefault}{\updefault}{\color[rgb]{0,0,0}$t$ in s}%
}}}}
\put(946,-1321){\rotatebox{90.0}{\makebox(0,0)[rb]{\smash{{\SetFigFont{20}{24.0}{\rmdefault}{\mddefault}{\updefault}{\color[rgb]{0,0,0}$\sqb{v_1}(t)$ in V}%
}}}}}
\put(1846,-5056){\makebox(0,0)[b]{\smash{{\SetFigFont{20}{24.0}{\rmdefault}{\mddefault}{\updefault}{\color[rgb]{0,0,0}$0$}%
}}}}
\put(2836,-5056){\makebox(0,0)[b]{\smash{{\SetFigFont{20}{24.0}{\rmdefault}{\mddefault}{\updefault}{\color[rgb]{0,0,0}$2$}%
}}}}
\put(3826,-5056){\makebox(0,0)[b]{\smash{{\SetFigFont{20}{24.0}{\rmdefault}{\mddefault}{\updefault}{\color[rgb]{0,0,0}$4$}%
}}}}
\put(4816,-5056){\makebox(0,0)[b]{\smash{{\SetFigFont{20}{24.0}{\rmdefault}{\mddefault}{\updefault}{\color[rgb]{0,0,0}$6$}%
}}}}
\put(5851,-5056){\makebox(0,0)[b]{\smash{{\SetFigFont{20}{24.0}{\rmdefault}{\mddefault}{\updefault}{\color[rgb]{0,0,0}$8$}%
}}}}
\put(6841,-5056){\makebox(0,0)[b]{\smash{{\SetFigFont{20}{24.0}{\rmdefault}{\mddefault}{\updefault}{\color[rgb]{0,0,0}$10$}%
}}}}
\put(1756,-4786){\makebox(0,0)[rb]{\smash{{\SetFigFont{20}{24.0}{\rmdefault}{\mddefault}{\updefault}{\color[rgb]{0,0,0}$0.0$}%
}}}}
\put(1756,-736){\makebox(0,0)[rb]{\smash{{\SetFigFont{20}{24.0}{\rmdefault}{\mddefault}{\updefault}{\color[rgb]{0,0,0}$2.0$}%
}}}}
\put(1756,-1501){\makebox(0,0)[rb]{\smash{{\SetFigFont{20}{24.0}{\rmdefault}{\mddefault}{\updefault}{\color[rgb]{0,0,0}$1.6$}%
}}}}
\put(1756,-2356){\makebox(0,0)[rb]{\smash{{\SetFigFont{20}{24.0}{\rmdefault}{\mddefault}{\updefault}{\color[rgb]{0,0,0}$1.2$}%
}}}}
\put(1756,-3166){\makebox(0,0)[rb]{\smash{{\SetFigFont{20}{24.0}{\rmdefault}{\mddefault}{\updefault}{\color[rgb]{0,0,0}$0.8$}%
}}}}
\put(1756,-3976){\makebox(0,0)[rb]{\smash{{\SetFigFont{20}{24.0}{\rmdefault}{\mddefault}{\updefault}{\color[rgb]{0,0,0}$0.4$}%
}}}}
\end{picture}%

%% file: fig_SNR/fig_VCPE_large.pstex_t
\begin{picture}(0,0)%
\includegraphics{fig_SNR/fig_VCPE_large}%
\end{picture}%
\setlength{\unitlength}{4144sp}%
\begingroup\makeatletter\ifx\SetFigFont\undefined%
\gdef\SetFigFont#1#2#3#4#5{%
  \reset@font\fontsize{#1}{#2pt}%
  \fontfamily{#3}\fontseries{#4}\fontshape{#5}%
  \selectfont}%
\fi\endgroup%
\begin{picture}(6819,5406)(619,-5590)
\put(6166,-5476){\makebox(0,0)[rb]{\smash{{\SetFigFont{20}{24.0}{\rmdefault}{\mddefault}{\updefault}{\color[rgb]{0,0,0}$t$ in s}%
}}}}
\put(946,-1321){\rotatebox{90.0}{\makebox(0,0)[rb]{\smash{{\SetFigFont{20}{24.0}{\rmdefault}{\mddefault}{\updefault}{\color[rgb]{0,0,0}$\sqb{v_1}(t)$ in V}%
}}}}}
\put(1846,-5056){\makebox(0,0)[b]{\smash{{\SetFigFont{20}{24.0}{\rmdefault}{\mddefault}{\updefault}{\color[rgb]{0,0,0}$0$}%
}}}}
\put(2836,-5056){\makebox(0,0)[b]{\smash{{\SetFigFont{20}{24.0}{\rmdefault}{\mddefault}{\updefault}{\color[rgb]{0,0,0}$2$}%
}}}}
\put(3826,-5056){\makebox(0,0)[b]{\smash{{\SetFigFont{20}{24.0}{\rmdefault}{\mddefault}{\updefault}{\color[rgb]{0,0,0}$4$}%
}}}}
\put(4816,-5056){\makebox(0,0)[b]{\smash{{\SetFigFont{20}{24.0}{\rmdefault}{\mddefault}{\updefault}{\color[rgb]{0,0,0}$6$}%
}}}}
\put(5851,-5056){\makebox(0,0)[b]{\smash{{\SetFigFont{20}{24.0}{\rmdefault}{\mddefault}{\updefault}{\color[rgb]{0,0,0}$8$}%
}}}}
\put(6841,-5056){\makebox(0,0)[b]{\smash{{\SetFigFont{20}{24.0}{\rmdefault}{\mddefault}{\updefault}{\color[rgb]{0,0,0}$10$}%
}}}}
\put(1756,-4786){\makebox(0,0)[rb]{\smash{{\SetFigFont{20}{24.0}{\rmdefault}{\mddefault}{\updefault}{\color[rgb]{0,0,0}$0.0$}%
}}}}
\put(1756,-736){\makebox(0,0)[rb]{\smash{{\SetFigFont{20}{24.0}{\rmdefault}{\mddefault}{\updefault}{\color[rgb]{0,0,0}$2.0$}%
}}}}
\put(1756,-1501){\makebox(0,0)[rb]{\smash{{\SetFigFont{20}{24.0}{\rmdefault}{\mddefault}{\updefault}{\color[rgb]{0,0,0}$1.6$}%
}}}}
\put(1756,-2356){\makebox(0,0)[rb]{\smash{{\SetFigFont{20}{24.0}{\rmdefault}{\mddefault}{\updefault}{\color[rgb]{0,0,0}$1.2$}%
}}}}
\put(1756,-3166){\makebox(0,0)[rb]{\smash{{\SetFigFont{20}{24.0}{\rmdefault}{\mddefault}{\updefault}{\color[rgb]{0,0,0}$0.8$}%
}}}}
\put(1756,-3976){\makebox(0,0)[rb]{\smash{{\SetFigFont{20}{24.0}{\rmdefault}{\mddefault}{\updefault}{\color[rgb]{0,0,0}$0.4$}%
}}}}
\end{picture}%

%% file: fig_SNR/fig_Vout_small.pstex_t
\begin{picture}(0,0)%
\includegraphics{fig_SNR/fig_Vout_small}%
\end{picture}%
\setlength{\unitlength}{4144sp}%
\begingroup\makeatletter\ifx\SetFigFont\undefined%
\gdef\SetFigFont#1#2#3#4#5{%
  \reset@font\fontsize{#1}{#2pt}%
  \fontfamily{#3}\fontseries{#4}\fontshape{#5}%
  \selectfont}%
\fi\endgroup%
\begin{picture}(6819,5406)(619,-5590)
\put(6166,-5476){\makebox(0,0)[rb]{\smash{{\SetFigFont{20}{24.0}{\rmdefault}{\mddefault}{\updefault}{\color[rgb]{0,0,0}$t$ in s}%
}}}}
\put(946,-1321){\rotatebox{90.0}{\makebox(0,0)[rb]{\smash{{\SetFigFont{20}{24.0}{\rmdefault}{\mddefault}{\updefault}{\color[rgb]{0,0,0}$\sqb{v}(t)$ in V}%
}}}}}
\put(1846,-5056){\makebox(0,0)[b]{\smash{{\SetFigFont{20}{24.0}{\rmdefault}{\mddefault}{\updefault}{\color[rgb]{0,0,0}$0$}%
}}}}
\put(2836,-5056){\makebox(0,0)[b]{\smash{{\SetFigFont{20}{24.0}{\rmdefault}{\mddefault}{\updefault}{\color[rgb]{0,0,0}$2$}%
}}}}
\put(3826,-5056){\makebox(0,0)[b]{\smash{{\SetFigFont{20}{24.0}{\rmdefault}{\mddefault}{\updefault}{\color[rgb]{0,0,0}$4$}%
}}}}
\put(4816,-5056){\makebox(0,0)[b]{\smash{{\SetFigFont{20}{24.0}{\rmdefault}{\mddefault}{\updefault}{\color[rgb]{0,0,0}$6$}%
}}}}
\put(5851,-5056){\makebox(0,0)[b]{\smash{{\SetFigFont{20}{24.0}{\rmdefault}{\mddefault}{\updefault}{\color[rgb]{0,0,0}$8$}%
}}}}
\put(6841,-5056){\makebox(0,0)[b]{\smash{{\SetFigFont{20}{24.0}{\rmdefault}{\mddefault}{\updefault}{\color[rgb]{0,0,0}$10$}%
}}}}
\put(1756,-4786){\makebox(0,0)[rb]{\smash{{\SetFigFont{20}{24.0}{\rmdefault}{\mddefault}{\updefault}{\color[rgb]{0,0,0}$2.0$}%
}}}}
\put(1756,-3841){\makebox(0,0)[rb]{\smash{{\SetFigFont{20}{24.0}{\rmdefault}{\mddefault}{\updefault}{\color[rgb]{0,0,0}$2.5$}%
}}}}
\put(1756,-2941){\makebox(0,0)[rb]{\smash{{\SetFigFont{20}{24.0}{\rmdefault}{\mddefault}{\updefault}{\color[rgb]{0,0,0}$3.0$}%
}}}}
\put(1756,-1996){\makebox(0,0)[rb]{\smash{{\SetFigFont{20}{24.0}{\rmdefault}{\mddefault}{\updefault}{\color[rgb]{0,0,0}$3.5$}%
}}}}
\put(1756,-1051){\makebox(0,0)[rb]{\smash{{\SetFigFont{20}{24.0}{\rmdefault}{\mddefault}{\updefault}{\color[rgb]{0,0,0}$4.0$}%
}}}}
\end{picture}%

%% file: fig_SNR/fig_Vout_large.pstex_t
\begin{picture}(0,0)%
\includegraphics{fig_SNR/fig_Vout_large}%
\end{picture}%
\setlength{\unitlength}{4144sp}%
\begingroup\makeatletter\ifx\SetFigFont\undefined%
\gdef\SetFigFont#1#2#3#4#5{%
  \reset@font\fontsize{#1}{#2pt}%
  \fontfamily{#3}\fontseries{#4}\fontshape{#5}%
  \selectfont}%
\fi\endgroup%
\begin{picture}(6819,5406)(619,-5590)
\put(6166,-5476){\makebox(0,0)[rb]{\smash{{\SetFigFont{20}{24.0}{\rmdefault}{\mddefault}{\updefault}{\color[rgb]{0,0,0}$t$ in s}%
}}}}
\put(946,-1321){\rotatebox{90.0}{\makebox(0,0)[rb]{\smash{{\SetFigFont{20}{24.0}{\rmdefault}{\mddefault}{\updefault}{\color[rgb]{0,0,0}$\sqb{v}(t)$ in V}%
}}}}}
\put(1846,-5056){\makebox(0,0)[b]{\smash{{\SetFigFont{20}{24.0}{\rmdefault}{\mddefault}{\updefault}{\color[rgb]{0,0,0}$0$}%
}}}}
\put(2836,-5056){\makebox(0,0)[b]{\smash{{\SetFigFont{20}{24.0}{\rmdefault}{\mddefault}{\updefault}{\color[rgb]{0,0,0}$2$}%
}}}}
\put(3826,-5056){\makebox(0,0)[b]{\smash{{\SetFigFont{20}{24.0}{\rmdefault}{\mddefault}{\updefault}{\color[rgb]{0,0,0}$4$}%
}}}}
\put(4816,-5056){\makebox(0,0)[b]{\smash{{\SetFigFont{20}{24.0}{\rmdefault}{\mddefault}{\updefault}{\color[rgb]{0,0,0}$6$}%
}}}}
\put(5851,-5056){\makebox(0,0)[b]{\smash{{\SetFigFont{20}{24.0}{\rmdefault}{\mddefault}{\updefault}{\color[rgb]{0,0,0}$8$}%
}}}}
\put(6841,-5056){\makebox(0,0)[b]{\smash{{\SetFigFont{20}{24.0}{\rmdefault}{\mddefault}{\updefault}{\color[rgb]{0,0,0}$10$}%
}}}}
\put(1756,-4786){\makebox(0,0)[rb]{\smash{{\SetFigFont{20}{24.0}{\rmdefault}{\mddefault}{\updefault}{\color[rgb]{0,0,0}$2.0$}%
}}}}
\put(1756,-3841){\makebox(0,0)[rb]{\smash{{\SetFigFont{20}{24.0}{\rmdefault}{\mddefault}{\updefault}{\color[rgb]{0,0,0}$2.5$}%
}}}}
\put(1756,-2941){\makebox(0,0)[rb]{\smash{{\SetFigFont{20}{24.0}{\rmdefault}{\mddefault}{\updefault}{\color[rgb]{0,0,0}$3.0$}%
}}}}
\put(1756,-1996){\makebox(0,0)[rb]{\smash{{\SetFigFont{20}{24.0}{\rmdefault}{\mddefault}{\updefault}{\color[rgb]{0,0,0}$3.5$}%
}}}}
\put(1756,-1051){\makebox(0,0)[rb]{\smash{{\SetFigFont{20}{24.0}{\rmdefault}{\mddefault}{\updefault}{\color[rgb]{0,0,0}$4.0$}%
}}}}
\end{picture}%